%% file: main.tex
\documentclass[11pt,letterpaper]{amsart}
\usepackage{mathrsfs, amssymb,amsmath,url,amsthm,color,comment}
\usepackage{hyperref}
\usepackage{xspace}
\usepackage{enumitem}
\usepackage[dvipsnames]{xcolor}
\usepackage{thmtools} 
\usepackage{thm-restate}
\usepackage[letterpaper,left=1in,top=1in,right=1in,bottom=1in]{geometry}

\author
{Antoine Mottet}
	\address{Hamburg University of Technology, Research Group on Theoretical Computer Science, Germany}
	\email{antoine.mottet@tuhh.de}
	\urladdr{http://amottet.github.io/~mottet/}
	
\author
{Tom\'{a}\v{s} Nagy}
	\address{Institut f\"{u}r Diskrete Mathematik und Geometrie, FG Algebra, TU Wien, Austria \and Theoretical Computer Science Department, Jagiellonian University, Poland
 }
	\email{tomas.nagy@email.com}
	\urladdr{https://nagyto.github.io/}

\author
{Michael Pinsker}
	\address{Institut f\"{u}r Diskrete Mathematik und Geometrie, FG Algebra, TU Wien, Austria}    
	\email{marula@gmx.at}
    \urladdr{http://dmg.tuwien.ac.at/pinsker/}
    
\author
{Micha{\l} Wrona}
	\address{Theoretical Computer Science Department, Jagiellonian University, Poland}
	\email{michal.wrona@uj.edu.pl}
	\urladdr{https://www.tcs.uj.edu.pl/wrona}

\thanks{Antoine Mottet has received funding from the European Research Council
(ERC) under the European Union's Horizon 2020 research and
innovation programme (grant agreement No 771005). Michael Pinsker and Tom\'{a}\v{s} Nagy have received partial funding from the Czech Science Foundation (grant No 18-20123S). Micha{\l} Wrona is partially supported by National Science Centre, Poland grant number 2020/37/B/ST6/01179.  This research was funded in whole or in part by the Austrian Science
Fund (FWF) [M 2555-N35, P 32337, I 5948]. For the purpose of Open Access, the authors have
applied a CC BY public copyright licence to any Author Accepted
Manuscript (AAM) version arising from this submission. This research is also funded by the European Union (ERC,
  POCOCOP, 101071674). Views and opinions expressed are however those
  of the author(s) only and do not necessarily reflect those of the
  European Union or the European Research Council Executive Agency.
  Neither the European Union nor the granting authority can be held
  responsible for them.}

\title[When symmetries are enough]{Collapsing the bounded width hierarchy\\ for infinite-domain  CSPs: \\ when symmetries are enough  \footnote{\normalfont\scshape\lowercase{A conference version of this article appeared in the Proceedings of the 48th International Colloquium on Automata, Languages and Programming, pages 1--20, 2021, under the title ``Smooth Approximations and Relational Width Collapses".}}
}

\usepackage{hyperref,multirow}

\usepackage{graphicx}

\input{commands}

\begin{document}

\begin{abstract}
    We prove that relational structures admitting specific polymorphisms (namely, canonical pseudo-WNU  operations of all arities $n \geq 3$) have low relational width.
    This implies a collapse of the bounded width hierarchy for numerous classes of infinite-domain CSPs studied in the literature.
    Moreover, we obtain a characterization of bounded width for first-order reducts of unary structures and a characterization of MMSNP sentences that are equivalent to a Datalog program, answering a question posed by Bienvenu \emph{et al.}.
    In particular, the bounded width hierarchy collapses in those cases as well.
    Our results extend the scope of theorems of Barto and Kozik characterizing bounded width  for finite structures, and show the applicability of   infinite-domain CSPs to other fields.
\end{abstract}

\maketitle

\section{Introduction}\label{sect:intro}

Understanding what mathematical structure in problems enables efficient algorithms is one of the most fundamental goals in theoretical computer science; the field therefore strives to achieve such understanding completely  for 
restricted classes of problems that are  sufficiently general to contribute to the overall picture. Another  core goal is   identifying  large natural subclasses of the complexity class NP that exhibit a complexity dichotomy, i.e., which contain only problems that are either solvable in polynomial time or NP-complete, and therefore avoid intermediate complexity classes. One class of problems that touches both of the above directions is the class of \emph{constraint satisfaction problems} (CSPs), which can be cast as homomorphism problems (of a given structure to a fixed structure called a \emph{template}). Feder and Vardi famously proposed the conjecture that the class of CSPs with a finite template  exhibits such a dichotomy~\cite{FederVardi}. With time, this goal evolved from just proving the dichotomy to precisely explaining, as uniformly as possible, the mathematical structure of problems that causes their tractability or hardness. One reason why this goal received considerable attention and was pursued for decades is that CSPs with a finite template hit the sweet spot between generality and structure -- they are  quite general and include many computational  problems of interest, yet they posses sufficient structure that, in a highly non-trivial way,  can be extracted and either be used algorithmically or in hardness proofs.

The Feder-Vardi conjecture having been confirmed independently by Bulatov and Zhuk in 2017~\cite{Bulatov:2017,Zhuk:2017,Zhuk:2020}, it is natural to look for a more general class of problems where potentially there is still sufficient  structure to explain all tractability within the class and that possibly exhibits a similar complexity dichotomy.
There are several ways to expand the framework, but one of the most natural ways is to allow infinite templates, thereby including many natural problems that can be formulated as CSPs, but not with a finite template.
This is the case for linear programming, some reasoning problems in artificial intelligence such as ontology-mediated querying, or even problems as simple to formulate as the digraph acyclicity problem.
On the other hand, allowing arbitrary infinite templates would amount to including basically any computational problem~\cite{BodirskyGrohe} and defy the hope for a structural approach as well as for a complexity dichotomy (or even containment in NP); it is therefore necessary to restrict the class of infinite templates under consideration.
To achieve the objective of a restricted class that does not suffer from the above deficiencies, properly extends the class of finite-template CSPs, and enriches the latter class by interesting problems,  model theory provides us with a useful language and methodology -- it has been used to identify a class of infinite templates which share some properties with finite structures
such that the corresponding CSPs are always in NP and such that a structural approach to explain tractability within the class is, in principle, achievable. 
A complexity  dichotomy for this class, whose templates are \emph{first-order reducts of $k$-homogeneous $\ell$-bounded $\omega$-categorical structures} (explained on a high level in Section~\ref{subsect:results}, and defined precisely in Section~\ref{sect:prelims}), has been conjectured by Bodirsky and Pinsker~\cite{BPP-projective-homomorphisms,Topo,BartoPinskerDichotomy,BKOPP,BKOPP-equations}, and a rich mathematical theory to investigate the complexity of problems within the class has evolved over the years (see e.g.~\cite{BodirskyBook, infinitesheep}).

Methodologically, the key to resolving the Feder-Vardi dichotomy conjecture was the universal-algebraic approach. This approach assigns to each template the set of certain finite-dimensional symmetries, called \emph{polymorphisms}, and uses them to identify the structure in problem instances which is then used either in polynomial-time algorithms or in hardness proofs. Given how well-developed and powerful this approach is, it is natural to try to lift to the infinite-template framework as many tools from the finite case as possible. The most general results in the area of infinite-template CSPs obtained so far use \emph{canonical polymorphisms}, which behave like operations on finite sets and which sometimes allow to mimic the universal-algebraic approach to finite-template CSPs. Canonical polymorphisms alone are in general insufficient to explain, say, polynomial-time tractablility, and a characterisation of the range of their  applicability is  missing, but on the positive side their use  covers a vast majority of the results obtained in the area. Moreover, even in situations where canonical polymorphisms  do not yield a black-box argument by directly  lifting results for   finite-template CSPs, they prove useful in the design of more elaborate arguments~\cite{SmoothApproximations,hypergraphs}.

The two main algorithmic techniques employed in polynomial-time algorithms for finite-template CSPs are local consistency checking and generalised Gaussian elimination. The applicability of pure forms of these techniques is well understood.
Again, it is natural to understand the applicability of these techniques for infinite-template CSPs.
The present paper is focused on the local consistency checking technique, more precisely on the concept of \emph{bounded width}.
One way to see it is as a generalisation of bounded width resolution for SAT, but adapted to deal with general constraints, rather than just with clauses.
Intuitively speaking, local consistency checking consists in propagating local information through a structure so as to infer global information: consider, for example, computing the transitive closure of a relation as deriving global information from local one.

The use of local consistency methods is not limited to constraint satisfaction.
Indeed, local consistency checking is also used for such problems as the graph isomorphism problem, where it is known as the Weisfeiler-Leman algorithm. 
Here, the technique can be used to derive implied constraints that an isomorphism between two graphs has to satisfy so as to narrow down the search space, but local consistency is in fact powerful enough to solve the graph isomorphism problem over any non-trivial minor-closed class of graphs~\cite{Grohe:2012}.
Notably, the best algorithm for graph isomorphism to date also uses local consistency as a subroutine~\cite{Babai:2016}.
Finally, local consistency can be used to solve games involved in formal verification such as parity games and mean-payoff games~\cite{Mottet:2018}.
However, unlike in the other applications mentioned above, these algorithms are not efficient since the constructed instances are exponentially large.

One of the reasons for the ubiquity of local consistency methods is that their
 underlying principles can be described in many different languages, such as the language of category theory~\cite{Dawar:2017}, the language of finite model theory (by Spoiler-Duplicator games~\cite{Kolaitis:2000} or by homomorphism duality~\cite{Atserias:2007}), and logical definability (in Datalog, or infinitary logics with bounded number of variables)~\cite{FederVardi}. For constraint satisfaction problems over a finite template, the power of local consistency checking can 
be characterised 
by conditions on the set of
polymorphisms of the template: local consistency checking correctly solves 
the problem
if, and only if, the polymorphisms satisfy these conditions~\cite{BartoKozikBoundedWidth,MarotiMcKenzie,Maltsev-Cond}. 
Moreover, whenever 
this is the case, then in fact only a very restricted form of local consistency checking is needed~\cite{BartoCollapse}. This fact is known as the \emph{collapse of the bounded width hierarchy}, and it has strong consequences both for complexity and logic.
Firstly, the collapse provides efficient algorithms that are able to solve \emph{all} the CSPs that are solvable by local consistency methods~\cite{Kozik:2016}, which yields in particular a polynomial-time algorithm solving instances of the \emph{uniform CSP}  where the template is also part of the input. Secondly, this collapse induces collapses in all the areas mentioned at the beginning of this paragraph.

Unlike for finite-template CSPs, an algebraic characterisation of the applicability of local consistency checking is missing for infinite-template CSPs.
In fact, even in situations where the complexity of a CSP is captured by the polymorphisms of the template, no purely algebraic description of CSPs solvable by local consistency methods is possible~\cite{TopoRelevant,TopoRelevant-TAMS,Hrushovski-monsters}; this is even the case for well-behaved classes of problems such as temporal CSPs~\cite{Rydval:2020}, which belong to the range of the aforementioned Bodirsky-Pinsker  conjecture.
These negative results are to be compared with the recent result by Mottet and Pinsker~\cite{SmoothApproximations} that did provide an algebraic description for several subclasses of infinite-template CSPs within that range

\subsection{Results}\label{subsect:results}

In the present paper, we focus on applying the theory of canonical functions to study the power of local consistency checking for CSPs with  
templates  within the Bodirsky-Pinsker conjecture, namely, first-order reducts  of  $k$-homogeneous $\ell$-bounded $\omega$-categorical structures.
Our objective is threefold: firstly,
we wish  to obtain generic sufficient conditions that imply that local consistency checking solves a given CSP; secondly, we aim at understanding the amount of locality needed for local consistency methods to solve the CSP, as measured by the so-called \emph{relational width}; and thirdly, we want to show that our sufficient conditions are actually also necessary for certain  interesting  subclasses. We now give a high-level presentation of our results; precise  definitions of all concepts mentioned here can  be found in the preliminaries.

In order to achieve the first two objectives, we prove that sufficiently locally consistent instances  can be turned into locally consistent instances of a finite-template CSP. Imposing  conditions on the canonical polymorphisms  of the original template, we can enforce this finite-template CSP to have bounded width. It then has relational width at most $(2,3)$ by~\cite{BartoCollapse}; application of this result in turn allows us to obtain a bound on the relational width of the original CSP which only depends on the two simple parameters $k,\ell$ of the  presentation of the template  by a $k$-homogeneous $\ell$-bounded structure. Building on the results for finite templates, we therefore obtain a collapse of the relational width reminiscent of the collapse in the finite case for all structures whose canonical polymorphisms satisfy our conditions.

Before  stating this result in full detail, let us comment on the nature of the range of the Bodirsky-Pinsker conjecture and its use in local consistency algorithms. A countable structure is $\omega$-categorical if its automorphism group is large in the sense that it acts with finitely many orbits on $n$-tuples for all $n\geq 1$: that is, the structure has up to automorphisms only finitely many distinct tuples of a fixed finite length. The additional conditions of $k$-homogeneity and $\ell$-boundedness state that moreover, the orbit of any $n$-tuple is uniquely determined by the orbits of its $k$-subtuples, and that conversely which combinations of orbits of $k$-tuples can appear on an $n$-tuple is a local property (measured by $\ell$). A first-order reduct of such a structure is any structure all of whose relations are unions of orbits. Local consistency algorithms, applied to an instance of such a template, keep lists of possible orbits of tuples of variables, reduce them according to the constraints, and then propagate information between these lists locally to make them consistent.
By $\omega$-categoricity, this procedure terminates (and does so in polynomial time).

\begin{restatable}{theorem}{canonicalbwidth}\label{thm:canonicalbwidth}
Let $k,\ell\geq 1$, and let $\rel A$ be a first-order reduct of a $k$-homogeneous $\ell$-bounded $\omega$-categorical structure $\rel B$.
\begin{itemize}
    \item If the clone of $\Aut(\rel B)$-canonical polymorphisms of $\rel A$ contains pseudo-WNU operations modulo $\overline{\Aut(\rel B)}$ of all arities $n \geq 3$, then $\rel A$ has relational width at most $(2k,\max(3k,\ell))$.
    \item If the clone of $\Aut(\rel B)$-canonical polymorphisms of $\rel A$ 
    contains pseudo-totally symmetric operations modulo $\overline{\Aut(\rel B)}$ of all arities, then $\rel A$ has relational width at most 
    $(k,\max(k+1,\ell))$ if $k >1$, and $(1,1)$ if $k=1$.
\end{itemize}
\end{restatable}

We remark that if $\rel A$ is a finite template, then it satisfies the conditions of Theorem~\ref{thm:canonicalbwidth}; moreover, 
the width obtained in the first item of Theorem~\ref{thm:canonicalbwidth} coincides with the width given by Barto's collapse result from~\cite{BartoCollapse}, the width obtained in the second item coincides with the width of structures with totally symmetric operations of all arities from~\cite{DalmauPearson,FederVardi}, and our algebraic conditions coincide with the known conditions for finite templates. We now provide the short argument for the reader familiar with the precise notions of the theorem; others might prefer to skip the remainder of this paragraph, and return to it later.  
Every finite structure $\rel A$ with domain $\{a_1,\dots,a_n\}$ is a first-order reduct of the structure $\rel B=(\{a_1,\dots,a_n\};\{a_1\},\dots,\{a_n\})$, which is $1$-homogeneous and $2$-bounded (and trivially $\omega$-categorical). For $1$-homogeneity, note first that the automorphism group $\Aut(\rel B)$ of $\rel B$ is trivial. Hence, the orbit of any tuple with respect to $\Aut(\rel B)$ consists only of  the tuple itself, and is clearly determined by the unary singleton relations that hold on the tuple. To see 2-boundedness, note that in order for any structure to embed into $\rel B$, it is sufficient to ensure that all elements of that structure are contained in precisely one unary relation, and that no two distinct elements are contained in the same unary relation. Finally, it also  follows from the triviality of $\Aut(\rel B)$ that any polymorphism of $\rel A$ is $\Aut(\rel B)$-canonical, making this condition void,  that any pseudo-WNU operation modulo $\overline{\Aut(\rel B)}$ is also a WNU operation, and that any pseudo-totally symmetric operation modulo $\overline{\Aut(\rel B)}$ is a totally symmetric operation.

As a corollary of Theorem~\ref{thm:canonicalbwidth}, we obtain a collapse of the bounded width hierarchy for first-order reducts of numerous structures studied in the literature~\cite{Schaefer-Graphs,ReductionFinite,HomogeneousGraphs,PosetCSP} where applicability of local consistency methods had already  been classified, as well as of \emph{unary structures,} i.e., structures over a finite unary signature~\cite{ReductsUnary}. For the latter class, in order to apply our theorem, we first need to show that all structures in the class with bounded width satisfy the assumptions of the theorem; this is one of our contributions towards the third objective mentioned in the beginning of this section. To this end, we build on  recent work by Mottet and Pinsker~\cite{SmoothApproximations} and expand the use of their \emph{smooth approximations} to fully suit \emph{equational  (non-)affineness}, which is roughly  the algebraic situation imposed by local consistency solvability. 
The main technical contribution is a new \emph{loop lemma} that exploits deep algebraic tools from the finite~\cite{BartoKozikCyclic} and, assuming the use of canonical functions is unfruitful, allows to obtain the existence of polymorphisms of every arity $n\geq 3$ and satisfying certain strong symmetry conditions. The study of unary structures is motivated on the one hand by the fact that its first-order reducts are a subclass of the class of first-order reducts of classical structures like $(\mathbb Z; +,\leq)$, $(\mathbb Q;+,\leq, 1)$, or $(\mathbb R; +,\times)$~\cite{ReductsUnary}, which are of particular interest in mathematics and computer science~\cite{Mottet:2018, CSPs_RQ}. On the other hand, they are also a superclass of the class of finite structures, consisting of templates where single values can be blown up to finite or infinite sets (by means of unary predicates representing them). The class therefore constitutes an excellent testbed for algebraic and combinatorial generalizations, in the spirit of the introduction, of statements about finite structures to structures that are in a sense ``just beyond finite.''

\begin{restatable}{corollary}{explicitcollapse}\label{cor:explicitcollapse}
    Let $\rel A$ be a structure that has bounded width.
    If $\rel A$ is a first-order reduct of:
    \begin{itemize}
        \item the universal homogeneous graph $\rel{G}$ or tournament $\rel{T}$, or of a unary structure, then $\rel A$ has relational width at most $(4,6)$;
        \item the universal homogeneous $K_n$-free graph $\rel{H}_n$, where $n\geq 3$, then at most $(2,n)$;
        \item $(\rel{N}; =)$, the  countably infinite equivalence relation $\rel{C}^{\omega}_{\omega}$ with infinitely many equivalence classes, all of infinite size,  or the universal homogeneous  partial order $\rel{P}$,  then at most $(2,3)$. 
    \end{itemize}
\end{restatable}

In Section~\ref{sect:corollary}, after our detailed treatment of unary structures, we provide the proof of Corollary~\ref{cor:explicitcollapse} along with examples showing that most of the bounds from the statement are tight.

Our second  contribution (after unary structures) to the third objective of characterizing templates with bounded width concerns templates representing   the model-checking problem for Feder and Vardi's logic MMSNP~\cite{FederVardi}. Among CSPs definable in this logic, we characterize those with bounded width as precisely those whose appropriate template satisfies the conditions of Theorem~\ref{thm:canonicalbwidth}, again using smooth approximations and our new loop lemma. CSPs with a template having bounded width are definable in the declarative language Datalog, which is from a logical point of view the existential positive fragment of least fixpoint logic (see, e.g., \cite{FederVardi} for a clear introduction of Datalog in the context of constraint satisfaction).
Coupled with our results, this observation allows us to solve the following open problem from~\cite{Bienvenu:2014,FeierKuusistoLutz}. The \emph{Datalog-rewritability problem} for MMSNP is the problem of deciding, given as input an MMSNP sentence $\Phi$, whether $\neg\Phi$ is equivalent to a Datalog program.
This problem appears naturally in the area of ontology-mediated querying, where a certain querying language used in practice and known as $(\mathcal{ALC},\text{UCQ})$ is logically equivalent to the logic MMSNP.
The evaluation of such queries is in general a coNP-complete problem.
When such a query is Datalog-rewritable, the evaluation of the query becomes a polynomial-time tractable problem.
\begin{restatable}{theorem}{boundedwidthmmsnp}
\label{thm:datalogrewritability}
    The Datalog-rewritability problem for MMSNP is decidable, and is 2NExpTime-complete.
\end{restatable}

\subsection{Related results}

Local consistency for $\omega$-categorical structures was studied for the first time in~\cite{BodirskyDalmau} where basic notions were introduced and some basic results provided. Bounded width was characterised for first-order reducts of certain $k$-homogeneous $\ell$-bounded structures 
in~\cite{SmoothApproximations,Rydval:2020}.

A structure $\rel A$ has \emph{bounded strict width}~\cite{FederVardi}
if not only $\Csp(\rel A)$ is solvable by local consistency methods, but  moreover every partial solution of a locally consistent instance can be extended to a total solution.
The articles~\cite{Wrona:2020a} and~\cite{Wrona:2020b} give the upper bound $(2, \max(3,\ell))$ on the relational width for first-order expansions of some classes of $2$-homogeneous, $\ell$-bounded structures under the stronger assumption of bounded strict width; Corollary~\ref{cor:explicitcollapse} for first-order reducts of $\rel{H}_n$ and of $\rel{C}^{\omega}_{\omega}$ also follows from~\cite{Wrona:2020a}.  The results of~\cite{Wrona:2020b} were generalized in~\cite{strictwidth}, where the upper bound $(k,\max(k+1,\ell))$ on the relational width was given for first-order expansions of certain classes of $k$-homogeneous, $\ell$-bounded structures with bounded strict width for $k>2$.

It has been recently shown in~\cite{Wrona:2024} that for a large class of $\omega$-categorical structures the existence of an \emph{oligopotent quasi near-unanimity polymorphism}, which is equivalent to having bounded strict width, and the conditions on canonical polymorphisms in Theorem~\ref{thm:canonicalbwidth} are not the only sufficient conditions for having bounded width. 
Namely, for certain $2$-homogeneous, $\ell$-bounded structures having a chain of 
\emph{directed quasi Jonsson polymorphisms} the bound of $(2, \max(3,\ell)))$ was obtained for the relational width, 
which is the same upper bound as in~\cite{Wrona:2020b}.

There are also classes of $k$-homogeneous $\ell$-bounded structures for which it is known that all their first-order reducts with bounded width have relational width at most $(2k,\max(3k,\ell))$, but where this collapse  follow s neither from Theorem~\ref{thm:canonicalbwidth} nor from any other result mentioned above. For the first-order reducts of $(\mathbb Q;<)$, such a collapse follows from~\cite{Rydval:2020,NebelBueckert}, and for first-order reducts of certain $k$-uniform hypergraphs with $k\geq 3$, this collapse was proven in~\cite{hypergraphs} using a combination of Theorem~\ref{thm:canonicalbwidth} and additional work. To our knowledge, there is no class of $k$-homogeneous $\ell$-bounded structures where the bound on the relational width from the first item of Theorem~\ref{thm:canonicalbwidth} would be known not to apply.

\subsection{Organisation of the present  article}
In Section~\ref{sect:prelims} we provide the basic notions and definitions. The reduction to the finite using canonical functions which leads to the collapse of the bounded width hierarchy, the proof of Theorem~\ref{thm:canonicalbwidth}, is given in Section~\ref{sect:collapse}. We then extend the algebraic theory of smooth approximations in Section~\ref{sect:looplemma} before applying it to first-order reducts of unary structures 
and MMSNP in Section~\ref{sect:newboundedwidth}. In this section we also give a proof of Corollary~\ref{cor:explicitcollapse} and Theorem~\ref{thm:datalogrewritability}.

\section{Preliminaries}\label{sect:prelims}

\subsection{Structures and model-theoretic notions}
All relational  structures in this article are assumed to be at most countable. 
For sets $B,I$, the \emph{orbit} of a tuple $b\in B^I$ under the action of a permutation group $\group$ on $B$
is the set $\{\alpha(b)\mid \alpha\in\group\}$; it is also called the \emph{$\group$-orbit} of $b$.
A countable relational  structure $\rel B$ is \emph{$\omega$-categorical} if its automorphism group  $\Aut(\rel B)$ is \emph{oligomorphic}, i.e., for all $n \geq 1$, the number of  orbits of the action of $\Aut(\rel B)$ on  $n$-tuples is finite.
For an $\omega$-categorical relational structure $\rel B$ and for every $n\geq 1$ it holds that two $n$-tuples of elements of the structure are in the same orbit under $\Aut(\rel B)$ if, and only if, they have the same \emph{type}, i.e., they satisfy the same first-order formulas. The \emph{atomic type} of an $n$-tuple is the set of all atomic formulas satisfied by this tuple; the structure $\rel B$ is \emph{homogeneous} if any two tuples of the same atomic type belong to the same orbit. For $k\geq 1$, we say that a relational structure $\rel B$ is \emph{$k$-homogeneous} if for all tuples $a,b$ of arbitrary finite equal length, if all $k$-subtuples of $a$ and $b$ are in the same orbit under $\Aut(\rel B)$, then $a$ and $b$ are in the same orbit under $\Aut(\rel B)$. 
For $\ell\geq 1$,  we say that $\rel B$ is \emph{$\ell$-bounded} if for every finite $\rel X$, if all substructures $\rel Y$ of $\rel X$ of size at most $\ell$ embed in $\rel B$, then $\rel X$ embeds in $\rel B$. The structure $\rel B$ is \emph{universal} for a class of finite structures in its  signature if it embeds all members of the class. There exist an up to isomorphism unique countable universal homogeneous graph, partial order, tournament, and $K_n$-free graph (i.e., for any fixed $n\geq 3$, a graph not containing any  complete graph on $n$-vertices), respectively. 
We say that $\rel B$ is \emph{unary} if all relations in its signature are unary.
A \emph{first-order reduct} of a structure $\rel B$ is a structure on the same domain whose relations have a first-order definition without parameters in $\rel B$. 
\begin{example}[label=ex:Q]
The relational structure $(\Q;<)$ consisting of the set of rational numbers with the standard linear order has precisely $3$ orbits of pairs, namely the orbits $\{(a,b)\in\Q^2\mid a<b\}$, $\{(a,b)\in\Q^2\mid a>b\}$, $\{(a,b)\in\Q^2\mid a=b\}$; in the following, we will denote these orbits by $<$, $>$, and $\textrm{Eq}$, respectively, by abuse of notation. The structure $\Aut(\Q;<)$ is $2$-homogeneous since for any $k\geq 2$, the orbit of a $k$-tuple of elements of $\Q$ is determined by the orbits of its subpairs. It immediately follows that $(\Q;<)$ is also homogeneous and that the automorphism group $\Aut(\Q;<)$ is oligomorphic. Moreover, the structure $(\Q;<)$ is universal for the class of all finite linear orders. Finally, $(\Q;<)$ is $3$-bounded since in order to check whether a finite structure $\rel A=(A;<)$ embeds into $(\Q;<)$, it is enough to check that $<$ is a linear order on $A$, i.e. that it is a total, antisymmetric, antireflexive, and transitive relation; these properties only depend on the substructures of $\rel A$ of size at most $3$. An interesting example of a first-order reduct of $(\Q;<)$ is the structure $(\Q;\textrm{Betw})$, where $\textrm{Betw}$ is the ternary betweenness relation defined by $\{(a,b,c)\in\Q^3\mid a<b<c\text{ or }a>b>c\}$.
\end{example}

A first-order formula  is called  \emph{primitive positive} (pp) 
if it is built  exclusively from  atomic formulae, existential quantifiers, and conjunction. A  relation is \emph{pp-definable} in a structure $\rel B$ if it is first-order definable by a pp-formula.

Two relational structures $\rel A$ and $\rel B$ are \emph{homomorphically equivalent} if there exist homomorphisms from $\rel A$ to $\rel B$ and from $\rel B$ to $\rel A$. We remark that it is easy to see that homomorphically equivalent structures have the same relational width (to be defined later).

A relational structure $\rel A$ is \emph{Ramsey} if for any finite substructures $\rel B,\rel C$, of $\rel A$, there exists a substructure $\rel D$ of $\rel A$ such that the following holds. For every mapping $\chi$ from the set of all embeddings of $\rel B$ into $\rel D$ to $[k]$, there exists an embedding $f$ of $\rel C$ into $\rel D$ such that for all embeddings $g_1,g_2$ of $\rel B$ into $\rel C$, $\chi(f\circ g_1)=\chi(f\circ g_2)$.

\subsection{Polymorphisms, clones and identities}\label{subsect:clones}

A \emph{polymorphism} of a relational structure $\rel A$
is a homomorphism from some finite power of $\rel A$ to $\rel A$. 
The set of all  polymorphisms of a structure $\rel A$ is denoted by $\Pol(\rel A)$; it is a \emph{function clone}, i.e., a set of finitary operations on a fixed set which contains all projections and
which is closed under arbitrary compositions.

If $\cC$ is a function clone, then we denote the domain of its functions by $C$,  i.e., for every function $f\in\clone$, there exists $k\geq 1$ such that $f\in C^{C^k}$; we say that $\clone$ \emph{acts on} $C$. 
The clone $\clone$ also naturally acts (componentwise) on $C^l$ for any $l\geq 1$, and we write $\clone \curvearrowright C^l$ for this action. 
We say that a relation $S\subseteq C^l$ is \emph{invariant} under $\clone$ if for every $f\in\clone$ of arity $k\geq 1$, and for all $a_1,\dots,a_k\in S$, it holds that $f(a_1,\dots,a_k)\in S$; in this expression, we are referring to the action of $\clone$ on $C^l$. If $S$ is invariant, $\clone$ acts naturally on $S$ by restriction, and we write $\clone \curvearrowright S$ for this action. Finally, if $\sim$ is an invariant equivalence relation on an invariant relation $S$, then $\clone$ acts naturally on the classes of $\sim$ (by its action on representatives of the classes); we write $\clone\curvearrowright S/{\sim}$. Any  action $\clone \curvearrowright S/{\sim}$ is called a \emph{subfactor} of $\cC$, and we also call the pair $(S,\sim)$ a subfactor. A  subfactor $(S,\sim)$ is \emph{minimal} if $\sim$ has at least  two classes and no proper subset of $S$ intersecting at least two $\sim$-classes is invariant under $\cC$. For a clone $\clone$ acting on a set $X$ and  $Y\subseteq X$ we write $\langle Y \rangle_{\clone}$ for the smallest $\clone$-invariant subset of $X$ containing $Y$.

Let $n,k\geq 1$, let $A, B$ be sets, and let $\group$ and $\group'$ be permutation groups on $A$ and $B$, respectively. A $k$-ary function $f\colon A^k\rightarrow B$ is \emph{$n$-canonical} from $\group$ to $\group'$ if for all $a_1,\dots,a_k\in A^n$ and all $\alpha_1,\dots,\alpha_k\in\group$ there exists $\beta\in\group'$ such that $f(a_1,\dots,a_k)=\beta\circ f(\alpha_1(a_1),\dots,\alpha_k(a_k))$ (in these expressions, we consider the componentwise action of $f$ on $A^n$ and of $\group,\group'$ on $A^n$ and $B^n$, respectively). We say that $f$ is \emph{canonical} from $\group$ to $\group'$ if it is $n$-canonical from $\group$ to $\group'$ for every $n\geq 1$. We say that it is \emph{diagonally canonical} if it satisfies the definition of canonicity in case  $\alpha_1 = \cdots = \alpha_k$.  If $A=B$, $\group=\group'$, and $f$ is $n$-canonical, canonical, or diagonally canonical from $\group$ to $\group$, we say that $f$ is \emph{$n$-canonical, canonical, or diagonally canonical}, respectively, \emph{with respect to $\group$.} A function canonical with respect to $\group$ is also called \emph{$\group$-canonical}. If $f$ is $n$-canonical with respect to $\group$, it induces an operation on the set $C^n/\group$ of  $\group$-orbits of $n$-tuples. If all functions of a function clone $\clone$ are $n$-canonical with respect to $\group$, then $\clone$ acts on $C^n/\group$ and we write $\clone^n/\group$ for this action;  if $\group$ is oligomorphic then $\clone^n/\group$ is a function clone on a finite
set.

\begin{example}[continues=ex:Q]
The binary minimum operation $\min$ on $\Q$ as well as any function $\lex\colon \Q^2\rightarrow \Q$ satisfying $\lex(a,b)<\lex(a',b')$ if, and only if, $a<a'$ or $a=a'$ and $b<b'$ are both polymorphisms of $(\Q;<)$. The operation $\lex$ is canonical with respect to $\Aut(\Q;<)$ -- using the notation for orbits of pairs under $\Aut(\Q;<)$ introduced above, it acts on these orbits as follows: $\lex(<,<)=\lex(<,\textrm{Eq})=\lex(<,>)=\lex(\textrm{Eq},<)=\;<, \lex(>,<)=\lex(>,\textrm{Eq})=\lex(>,>)=\lex(\textrm{Eq},>)=\;> , \lex(\textrm{Eq},\textrm{Eq})= \textrm{Eq}$. On the other hand, $\min$ is not canonical with respect to $\Aut(\Q;<)$. Indeed, $0<1<2<3$ but $\min(0,1)=\min(1,0)$ and $\min(0,1)<\min(3,2)$, and hence the expression $\min(<,>)$ is not well-defined. However, $\min$ is diagonally canonical with respect to $\Aut(\Q;<)$. Indeed, for any tuple $(a_1,\dots,a_4)\in\Q^4$, the atomic type of $(\min(a_1,a_2),\min (a_3,a_4))$ depends only on the atomic type of $(a_1,\dots,a_4)$; hence the homogeneity of $(\Q;<)$ implies that if $(b_1,\dots,b_4)\in \Q^4$ lies in the same $\Aut(\Q;<)$-orbit as $(a_1,\dots,a_4)$, then the orbits of $(\min(a_1,a_2),\min (a_3,a_4))$ and $(\min(b_1,b_2),\min (b_3,b_4))$ coincide.
\end{example}

Note that for any family of permutation groups on the same set there is a smallest permutation group containing each member of the family, obtained by closing the union over the family under composition. Since any function clone $\clone$ is closed under composition, application of this fact to the family of all permutation groups which are contained (as a subset of the unary functions of $\clone$) in $\clone$ yields the largest permutation group contained in $\clone$ which we denote by $\group_{\clone}$. 
We say that $\clone$ is oligomorphic if $\group_{\clone}$ is oligomorphic. For $n\geq 1$, the \emph{$n$-canonical (canonical) part} of $\clone$ is the clone of those functions of $\clone$ which are $n$-canonical (canonical) with respect to $\group_{\clone}$. We write $\cann{\clone}$ and $\can{\clone}$ for these sets, which form themselves function clones. In particular, we use this notation for polymorphism clones, writing  $\cann{\Pol(\rel A)}$ and $\can{\Pol(\rel A)}$.

For a set of finitary functions $\mathscr{F}$ over the same fixed set $C$ we write $\overline{\mathscr F}$ for the set of those functions $g$ such that for all finite subsets $F$ of $C$, there exists a function in $\mathscr{F}$ which agrees with $g$ on $F$. 
We say that $f$ \emph{locally interpolates} $g$ modulo $\group$,
where $f,g$ are $k$-ary functions and  $\group$ is a permutation group all of which act on the same domain, 
if $g\in \overline{\{\beta\circ f(\alpha_1, \ldots, \alpha_k)\mid\beta, \alpha_1, \ldots, \alpha_k \in \group\}}$. 
Similarly, we say  that $f$ \emph{diagonally interpolates} $g$ modulo $\group$ if $f$ locally interpolates $g$ modulo $\group$ 
with $\alpha_1 = \cdots = \alpha_k$.
If $\group$ is the automorphism group of a \emph{Ramsey structure} in the sense of~\cite{BodirskyRamsey},
then every function on its domain locally (diagonally)  interpolates a canonical (diagonally canonical)  function modulo $\group$~\cite{BPT-decidability-of-definability, BodPin-CanonicalFunctions}.

\begin{theorem}[Consequence of Theorem 5 from~\cite{BodPin-CanonicalFunctions}]\label{thm:canonical-homo}
    Let $\rel A$ be a countable homogeneous Ramsey structure, let $\rel B$ be an $\omega$-categorical structure, and let $h\colon A\to B$ be an arbitrary function.
    Then $h$ locally interpolates modulo $\Aut(\rel A)$ a function that is canonical from $\Aut(\rel A)$ to $\Aut(\rel B)$, and it diagonally interpolates modulo $\Aut(\rel A)$ a function that is diagonally canonical from $\Aut(\rel A)$ to $\Aut(\rel B)$.
\end{theorem}

In particular, if $f\colon A^n\to A$ is a polymorphism of a structure $\rel A'$ with domain $A$, and $\rel A$ is a Ramsey structure, then $f$ (diagonally) locally interpolates modulo $\Aut(\rel A)$ a polymorphism $g$ of $\rel A'$ that is (diagonally) canonical with respect to $\Aut(\rel A)$.
We say that a clone $\cD$ locally interpolates a clone $\cC$ modulo a permutation group $\group$ if for every $g \in \cD$ there exists $f \in \cC$ such that $g$ locally interpolates $f$ modulo $\group$.  

A clone $\clone$ is a \emph{model-complete core} if its unary functions are equal to $\overline{\gG_\cC}$, and a structure $\rel A$ is called a  model-complete core if its polymorphism clone is. A finite model-complete core is called a \emph{core}. For every $\omega$-categorical structure $\rel A$ there exists an up to isomorphism unique  structure that is a model-complete core and that is homomorphically equivalent to $\rel A$ by~\cite{cores}; we call this structure \emph{the model-complete core of $\rel A$.} Since, as remarked earlier, the relational width of a structure is invariant under homomorphic equivalence, one can in this context often work without loss of generality  with model-complete cores. In particular, one can then use the following  property which distinguishes model-complete cores from other structures: 
Any $\omega$-categorical model-complete core pp-defines all orbits of  $n$-tuples with respect to its own automorphism group, for all $n \geq 1$~\cite{BodirskyBook}. Since for any $\omega$-categorical structure $\rel B$, pp-definability of a relation is equivalent to its invariance under  $\Pol(\rel B)$~\cite{BodNesetril}, in the case of model-complete cores we get that all orbits of its  automorphism group are invariant under all polymorphisms.

An \emph{identity} is an equation of terms built from some functional symbols where all variables are  implicitly interpreted as being  universally quantified. For example, if $f$ is a $k$-ary function symbol for some $k\geq 1$, then  $f(x, \ldots , x) = x$ is an identity, and a $k$-ary function on some domain $C$ \emph{satisfies} the identity $f(x, \ldots , x) = x$ if this equation holds for all elements $x$ of $C$. Functions satisfying this particular  identity of the appropriate arity are called \emph{idempotent}; a function clone is
idempotent if all of its functions are.
A $k$-ary operation $w$ is called a \emph{weak near-unanimity
(WNU)} operation if it satisfies the set of identities
containing an equation for each pair of terms in $\{ w(x, \ldots , x, y), w(x,y, \ldots , x), \ldots, w(y, x, \ldots , x) \}$.
It is called \emph{totally symmetric} if
$w(x_1, \ldots, x_k) = w(y_1, \ldots , y_k)$ whenever $\{x_1, \ldots, x_k \} = \{ y_1, \ldots , y_k \}$. 
Each
set of identities also has a \emph{pseudo}-variant obtained by composing each term appearing in the
identities with a distinct new unary function symbol. For example, a ternary pseudo-WNU operation $f$ satisfies the identities: $e_1 \circ f(y,x,x) = e_2 \circ f(x,y,x)$,
$e_3 \circ f(y,x,x) = e_4 \circ f(x,x,y)$ and
$e_5 \circ f(x,y,x) = e_6 \circ f(x,x,y)$. A function clone $\mathscr C$ \emph{satisfies} a given set of identities if the function symbols which appear in the identities can be assigned functions in $\mathscr C$ in such a way that all identities in the set are satisfied by the assigned functions;  if $\mathscr U\subseteq \mathscr C$ is a set of unary functions, then $\mathscr C$  \emph{satisfies a set of pseudo-identities modulo $\mathscr U$} if it satisfies the identities in such a way that the new unary function symbols which were added to the original identities to obtain the pseudo-variant are assigned values in $\mathscr U$. For an example for the satisfaction of identities in a polymorphism clone, note that  the clone of polymorphisms of $(\Q;<)$ satisfies the $n$-ary WNU identities for every $n\geq 2$, as witnessed by the $n$-ary operation $\min(x_1,\min(\dots,\min(x_{n-1},x_{n})))$. The function clone of all finitary operations on a countable set which are injective or arise from an injection by the addition of dummy variables is easily seen to be the polymorphism clone of a suitable structure (see e.g.~\cite{BodChenPinsker} for such a structure); it is not hard to see that it satisfies the pseudo-identity $e_1\circ f(x,y)=e_2\circ f(y,x)$ modulo the set of unary injections, but that it does not contain any commutative operation.

An arity-preserving map $\xi\colon \clone \to \cD$ between function clones is called a \emph{clone homomorphism} if it preserves
projections, i.e., it maps every  projection in $\cC$ to the corresponding  projection in $\cD$,  and compositions, i.e., it satisfies $\xi(f \circ (g_1, \ldots, g_n)) = \xi(f) \circ (\xi(g_1), \ldots, \xi(g_n))$ for all $n, m \geq 1$ and all $n$-ary $f \in \cC$ and $m$-ary $g_1, \ldots, g_n \in \cC$. 
An arity-preserving map $\xi$ is a \emph{minion homomorphism} if $\xi(f \circ (\pi_{i_1}, \ldots, \pi_{i_n})) = \xi(f) \circ (\pi_{i_1}, \ldots, \pi_{i_n})$ for all $n, m \geq 1$ and all $n$-ary $f \in \cC$ and $m$-ary projections $\pi_{i_1}, \ldots, \pi_{i_n}$.
We say that a function clone  $\clone$ is
\emph{equationally trivial} if it has a clone homomorphism to the clone $\Proj$ of projections over the two-element domain,  and \emph{equationally non-trivial} otherwise.
We also say that $\clone$  is
\emph{equationally affine} if it has a clone
homomorphism to an \emph{affine clone}, i.e., a clone of affine maps over a finite module. 
It is known that an  idempotent function clone on a finite domain is either equationally affine or it contains WNU operations of all  arities $n \geq 3$ (\cite{MarotiMcKenzie} shows a slightly weaker statement; an argument  proving the precise statement made here is attributed to E.\ Kiss in~\cite[Theorem 2.8]{Maltsev-Cond}; a more recent proof of this statement  can be found in~\cite{StrongSubalgebras}). 
Similarly, if $\rel A$ is an  $\omega$-categorical model-complete core, then $\Pol(\rel A)^{\canonical}$ 
is either equationally affine, or it contains
pseudo-WNU  operations modulo $\overline{\Aut(\rel A)}$ of all 
arities $n \geq 3$ (see~\cite{BPP-projective-homomorphisms,SmoothApproximations} for the lift of the corresponding result from the finite). 

If $\cC$, $\cD$ are function clones and $\cD$ has a finite domain, then a clone or minion homomorphism
$\xi\colon \cC \to \cD$ is \emph{uniformly continuous} if for all $n \geq 1$ there exists a finite subset $F$ of $C^n$
such that $\xi(f) = \xi(g)$
for all $n$-ary $f, g \in \cC$ which agree on $F$.
If $C$ is finite, then clearly every clone or minion homomorphism $\cC\to\cD$ is uniformly continuous.

\subsection{CSP, Relational Width, Minimality}

A \emph{CSP instance} over a set $A$ is a pair $\mathcal \instance=(\V,\constraints)$, where $\V$ is a finite set of variables, and $\mathcal C$ is a finite set of \emph{constraints}; each constraint $C\in \mathcal C$ is a subset of $A^U$ for some non-empty  $U\subseteq \V$; the set $U$ is called the \emph{scope} of $C$. A map $f\colon \V\rightarrow A$ is a \emph{solution} of a CSP instance $\instance=(\V,\constraints)$ if for every constraint $C\in\constraints$ with scope $U\subseteq \V$, it holds that $f|_U\in C$. If $\rel A$ is a relational structure with domain $A$, then we say that $\instance$ is an \emph{instance of $\Csp(\rel A)$} if  every constraint  $C\in\mathcal{C}$ can be viewed as a relation pp-definable in $\rel A$ by totally ordering its scope $U$; more precisely, there exists an enumeration  $u_1,\ldots,u_k$ of the elements of $U$ and  a $k$-ary relation $R$ that is pp-definable in $\rel A$ such that for all $f\colon U\to A$ we have  
$f\in C\Leftrightarrow (f(u_1),\dots,f(u_k))\in R$.
A CSP  instance  is \emph{non-trivial} if it does not contain any empty constraint; otherwise, it is \emph{trivial}.
Given a constraint $C\subseteq A^U$ and $K\subseteq U$, the \emph{projection of $C$ onto $K$} is defined by $C|_K:=\{f|_K\colon f\in C\}$. We say that a constraint $C$ of a $\Csp$ instance over a set $A$ is \emph{preserved} by a function $f\colon A^n\rightarrow A$ if for all $g_1,\ldots,g_n\in C$ we have $f(g_1,\ldots,g_n)\in C$. 

\begin{definition}\label{def:minimality}
Let $1\leq m\leq n$. We say that an instance $\instance$ of $\Csp(\rel A)$ with variables $\V$ is \emph{$(m,n)$-minimal}  if  both of the following hold:
\begin{itemize}
\item every at most $n$-element subset of the variable set $\V$ is contained in the scope of some constraint in $\instance$;
\item for every at most $m$-element subset of variables $K\subseteq \V$ and any two constraints $C_1, C_2 \in \instance$ whose scopes contain $K$, the projections of $C_1$ and $C_2$ onto $K$ coincide.
\end{itemize}
We say that an instance $\instance$ is \emph{$m$-minimal} if it is $(m,m)$-minimal.
\end{definition}

\begin{example}[label=ex:minimality]
Let us consider the following instance $\instance=(\V,\constraints)$ of $\Csp(\Q;<)$. The set $\V$ of variables consists of elements $v_1,v_2,v_3$, and the set $\constraints$ consists of the constraints $C_{(1,2)}, C_{(2,3)}, C_{(3,1)}$, where $C_{i,j}:=\{f\in \Q^{\{v_i,v_j\}}\mid f(v_i)<f(v_j)\}$ for every $(i,j)\in\{(1,2),(2,3),(3,1)\}$. The instance $\instance$ has clearly no solution since every assignment $g\in \Q^{\{v_1,v_2,v_3\}}$ with the property that $g|_{\{v_i,v_j\}}\in C_{(i,j)}$ for every $(i,j)\in\{(1,2),(2,3),(3,1)\}$ satisfies both $g(v_1)<g(v_2)$ since $g|_{\{v_1,v_2\}}\in C_{(1,2)}$, and $g(v_1)>g(v_2)$ (by the transitivity of $<$). The instance $\instance$ is $(2,2)$-minimal since every pair of variables is contained in the scope of some constraint, and $C_{(i,j)}|_{v_i}=\Q^{v_i}$, $C_{(i,j)}|_{v_j}=\Q^{v_j}$ for every $(i,j)\in\{(1,2),(2,3),(3,1)\}$.
\end{example}

Let $1\leq m\leq n$, and let $\rel A$ be a relational structure. 
Clearly not every instance $\instance=(\V,\mathcal{C})$ of $\Csp(\rel A)$ is $(m,n)$-minimal.
However, every instance $\instance$ is \emph{equivalent} to an $(m,n)$-minimal instance $\instance'$ in the sense that $\instance$ and $\instance'$ have the same set of solutions.
In particular we have that  if $\instance'$ is trivial, then $\instance$ has no solutions.
Moreover, if $\rel A$ is $\omega$-categorical,
then the instance  $\instance'$ can be computed in  only polynomially many steps in the number of variables and constraints as follows.
First introduce a new constraint $A^L$ for every set $L\subseteq\V$  with at most $n$ elements to satisfy the first condition.
Then remove orbits with respect to $\Aut(\rel A)$ from  the constraints in the instance as long as the second condition is not satisfied.

To see that $\instance'$ is an instance of $\Csp(\rel A)$, observe that the relation $A^L$ is pp-definable in $\rel A$ for every $L\subseteq \V$ and hence, the instance created from $\instance$ by adding the new constraints is still an instance of $\Csp(\rel A)$. Moreover, when the algorithm removes orbits from some constraint, then the new constraint is   obtained by a pp-definition from two old constraints and hence, it is pp-definable in $\rel A$. Whence, the instance $\instance'$ is an instance of $\Csp(\rel A)$.

The algorithm just described  introduces at most $\sum\limits_{i=1}^n\binom{|\V|}{i}$ new constraints. 
Let $p$ denote the maximum of $n$ and the maximal arity of the  relations of $\rel A$. Then the algorithm uses only relations of arity at most $p$. Note that by the $\omega$-categoricity of $\rel A$,  every relation pp-definable in $\rel A$ is a finite  union of  orbits of tuples with respect to $\Aut(\rel A)$. Hence, since in every step 
the algorithm removes at least one orbit from some constraint and since the number of orbits that can be removed is polynomial in $|\V|+|\constraints|$ by the discussion above, the number of steps of the algorithm is polynomial in $|\V|+|\constraints|$.

\begin{definition}
Let $1\leq m\leq n$. 
A relational structure $\rel A$ has \emph{relational width $(m,n)$} if every non-trivial $(m,n)$-minimal instance of $\rel A$ has a solution. $\rel A$ has \emph{bounded width} if it has relational width $(m,n)$ for some $m,n$.
\end{definition}

\begin{example}[continues=ex:minimality]
    Let us consider the instance $\instance$ from Example~\ref{ex:minimality}. The algorithm above which produces a $(2,3)$-minimal instance $\instance'$  equivalent to the instance $\instance$ proceeds as follows. It first introduces a new constraint $C_{(1,2,3)}:=\Q^{\V}$ so as to satisfy the first item of Definition~\ref{def:minimality}. In order to satisfy the second item, the algorithm compares projections of constraints to subsets of $\V$ containing at most $2$ elements. By comparing the projections of $C_{(1,2,3)}$ and of $C_{(1,2)}$ to $\{v_1,v_2\}$, it removes from $C_{(1,2,3)}$ all maps $f$ which do not satisfy $f(v_1)<f(v_2)$. By considering the projections of $C_{(1,2,3)}$ and $C_{(2,3)}$ to $\{v_2,v_3\}$ and of $C_{(1,2,3)}$ and $C_{(3,1)}$ to $\{v_3,v_1\}$, it finally removes all maps from $C_{(1,2,3)}$. Now, comparing the projections of $C_{(1,2,3)}$ and of $C_{(i,j)}$ for every $(i,j)\in\{(1,2),(2,3),(3,1)\}$, the algorithm removes also all maps from $C_{(1,2)}$, $C_{(2,3)}$, and from $C_{(3,1)}$. The resulting instance $\instance'$ is then $(2,3)$-minimal but it contains only empty constraints, and is therefore in particular trivial. In fact, it can be shown that $(\Q;<)$ has relational width precisely $(2,3)$.
\end{example}

The following statement gives an overview of conditions characterizing relational structures with bounded width that are used in several proofs later.
Let $p\geq 2$ be a  prime number and let $R_0$ and $R_1$ be the relations defined by $\{(x,y,z)\in\rel Z_p\mid x+y+z=i\bmod p\}$ for $i\in\{0,1\}$.
We say that a structure $\rel B$ has a \emph{pp-construction} in a structure $\rel A$ if $\rel B$ is homomorphically equivalent to a structure with domain $A^n$, for some $n\geq 1$, whose relations are pp-definable in $\rel A$  (for this purpose, a $k$-ary relation on $A^n$ is regarded as a $kn$-ary relation on $A$).

\begin{theorem}\label{thm:characterization-bwidth}
Let $\rel A$ be an $\omega$-categorical relational structure. All the implications $(i)\Rightarrow (j)$ and $(i)\Leftrightarrow (i')$ in the list below hold for $1\leq i\leq j\leq 3$.
\begin{enumerate}
    \item\label{itm:width-2-3} $\rel A$ has relational width at most $(2,3)$.
    \item\label{itm:bw} $\rel A$ has bounded width.
    \item[(\ref*{itm:bw}')]\label{itm:datalog} The class of finite structures that do not have a homomorphism to $\rel A$ is definable by a Datalog program.
    \item\label{itm:minion-affine} $\Pol(\rel A)$ does not admit a uniformly continuous minion homomorphism to an affine clone.
    \item[(\ref*{itm:minion-affine}')]\label{itm:pp-construction-zp} $\rel A$ does not pp-construct $(\mathbb Z_p;R_0,R_1)$ for any prime number $p$.
    \setcounter{bwcounter}{\value{enumi}}
\end{enumerate}
\noindent If $\rel A$ is finite, then all these conditions are equivalent, and moreover equivalent to the following two statements. 
\begin{enumerate}
\setcounter{enumi}{\value{bwcounter}}
\item\label{itm:core-affine} The expansion $\rel B$ of the core of $\rel A$ by all unary singleton relations is such that $\Pol(\rel B)$ is not equationally affine.
    \item\label{itm:wnu} $\Pol(\rel A)$ contains WNU operations of all arities $n\geq 3$.
\end{enumerate}
\end{theorem}

The implication (\ref*{itm:width-2-3}) to (\ref*{itm:bw}) is trivial, 
the equivalence of (\ref*{itm:bw}) and (\ref*{itm:datalog}') follows from \cite[Theorem 23]{FederVardi} and \cite[Corollary 1]{DualitiesCSP}.
(\ref*{itm:bw}) implies (\ref*{itm:minion-affine}) by~\cite{LaroseZadori} and~\cite{wonderland},
and the equivalence of (\ref*{itm:minion-affine}) and (\ref*{itm:pp-construction-zp}') is a folklore consequence of~\cite{wonderland}.
For finite structures, the implication from (\ref*{itm:core-affine}) to (\ref*{itm:width-2-3}) was proven in~\cite{BartoCollapse}, (\ref*{itm:minion-affine}) implies (\ref*{itm:core-affine}) by~\cite{wonderland}, and  (\ref*{itm:core-affine}) is equivalent to (\ref*{itm:wnu}) by the results from~\cite{MarotiMcKenzie,Maltsev-Cond} (in fact, a slightly weaker statement is formulated there; the precise  statement made here    is attributed to E.\ Kiss in~\cite[Theorem 2.8]{Maltsev-Cond}, and another  proof can be found in~\cite{StrongSubalgebras}).

\begin{theorem}[\cite{DalmauPearson,FederVardi}]\label{thm:ts-finite}
Let $\rel A$ be a finite relational structure. $\rel A$ has relational width $(1,1)$ if, and only if, $\Pol(\rel A)$ contains totally symmetric operations of all arities.
\end{theorem}

\subsection{Smooth Approximations}

We are going to apply the fundamental theorem of smooth approximations~\cite{SmoothApproximations} to lift an action of a function clone to a larger clone.

\begin{definition}
\emph{(Smooth approximations)}
Let $A$ be a set, $n \geq 1$, and let $\sim$ be an equivalence relation on a subset $S$ of $A^n$.
We say that an equivalence relation $\eta$ on some set $S'$
with $S \subseteq S'$ \emph{approximates} $\sim$ if the restriction of $\eta$ to $S$ is a (possibly non-proper) refinement of $\sim$. We call $\eta$ an \emph{approximation} of $\sim$.

For a permutation group $\group$ acting on $A$ and leaving $\eta$ as well as each  $\sim$-class invariant, we say
that the approximation $\eta$ is \emph{smooth} if each equivalence class $C$ of $\sim$ intersects some equivalence class $C'$ of $\eta$ such that
$C \cap C'$ contains a $\group$-orbit.
\end{definition}

\begin{theorem}[The fundamental theorem of smooth approximations~\cite{SmoothApproximations}]\label{thm:fundamental}
Let $\cC \subseteq \cD$ be function clones on a set $A$, and let $\group$ be a permutation group on $A$ such that $\cD$ locally interpolates $\cC$ modulo $\group$.
Let $\sim$ be a $\cC$-invariant equivalence relation on $S \subseteq A$ with $\group$-invariant classes and finite index,
and $\eta$ be a $\cD$-invariant smooth approximation of $\sim$ with respect to $\group$. Then there exists a uniformly continuous minion
homomorphism from $\cD$ to $\cC \curvearrowright S/{\sim}$.
\end{theorem}

\section{Collapses in the Relational Width Hierarchy}
\label{sect:collapse}

We now introduce a construction for reducing CSPs over structures in the scope of Theorem~\ref{thm:canonicalbwidth} to CSPs over finite sets. The idea is to create from an instance another instance whose variables are subsets of variables which take orbits as  values, of which there is only a finite number thanks to oligomorphicity. This will enable us to prove Theorem~\ref{thm:canonicalbwidth}.

\begin{definition}\label{def:instance-fin}
    Let $\mathcal I=(\V,\mathcal C)$ be a CSP instance over a set $A$.
    Let $\group$ be a permutation group on $A$, let $k\geq 1$, and for every $K\in\binom{\V}{k}$, let $\mathcal{O}_{\group}^{K}$ be the set of \emph{$K$-orbits}, i.e., orbits of maps $f\colon K\to A$ under the natural action of $\group$.
    Let $\mathcal{O}_{\group,k}^{\V}:=\bigcup\limits_{K\in\binom{\V}{k}} \mathcal{O}_{\group}^{K}$ and let
    $\instance_{\group,k}$ be the following instance over $\mathcal O_{\group,k}^{\V}$:
    \begin{itemize}
    \item The variable set of $\instance_{\group,k}$ is the set $\binom{\V}{k}$ of $k$-element subsets of $\V$.
Every variable $K$ of $\instance_{\group,k}$ is meant to take a value in $\mathcal O_{\group}^{K}$.

\item For every constraint $C\subseteq A^U$ in $\mathcal I$, $\instance_{\group,k}$ contains the constraint $C_{\group,k}\subseteq \mathcal (\mathcal O_{\group,k}^{\V})^{\binom{U}{k}}$ defined by
\[ C_{\group,k} := \left\{ g\colon \binom{U}{k}\to\mathcal \mathcal O_{\group,k}^{\V} \mid \exists f\in C\;\; \forall K\in\binom{U}{k}\;\; (f|_K \in g(K))\right\}. \]
That is, the set of $k$-element subsets of $U$ is the scope of a constraint which allows precisely those assignments of orbits to these subsets which are naturally induced by the assignments allowed by $C$ for the variables in $U$. Note that for every $K\in\binom{\V}{k}$, for every constraint $C$ whose scope contains $K$ and for every $g\in C_{\group,k}$ we have $g(K)\in \mathcal{O}_{\group}^{K}$. 
\end{itemize}
\end{definition}
Observe that if $\instance$ is non-trivial, then so is $\instance_{\group,k}$.

\begin{restatable}{lemma}{akbkminimal}\label{lem:ak-bk-minimal}
Let $1\leq a\leq b$. If $\mathcal I$ is $(ak,bk)$-minimal, then $\instance_{\group,k}$ is $(a,b)$-minimal.
\end{restatable}

\begin{proof}
	Let $b'\leq b$, and let $K_1,\dots,K_{b'}\in\binom{\V}{k}$. Note that $U:=\bigcup_i K_i$ has size at most $bk$, and therefore there exists a set $W$ with $U\subseteq W\subseteq \V$ and a constraint $C\subseteq A^W$ in $\mathcal I$ since $\mathcal I$ is $(ak,bk)$-minimal.
	The scope of the associated constraint $C_{\group,k}$ is $\binom{W}{k}$, which contains $K_1,\dots,K_{b'}$. Hence, the first item of Definition~\ref{def:minimality} is satisfied.
	
	For the second item, let $a'\leq a$, let $K:=\{K_1,\dots,K_{a'}\}$, and let $C_{\group,k}\subseteq (\mathcal O_{\group,k}^{\V})^{\binom{U}{k}}, D_{\group,k}\subseteq (\mathcal O_{\group,k}^{\V})^{\binom{W}{k}}$ be two constraints whose scopes contain $K$.
	Then $\bigcup K$ is contained in the scope of the associated $C\subseteq A^U$ and $D\subseteq A^W$ and has size at most $ak$, 
	so that by $(ak,bk)$-minimality of $\mathcal I$, the projections of $C$ and $D$ onto $\bigcup K$ coincide.
	Thus for every $g\in C_{\group,k}$, there exists by definition an $f\in C$ such that $f|_{K_i}\in g(K_i)$ for all $i$, and by the previous sentence  there exists  $f'\in D$ such that $f'|_{K_i}\in g(K_i)$ for all $i$.
	Thus, $g|_K$ is in the projection of $D_{\group,k}$ to $K$.
	The argument is symmetric, showing that the projections of $C_{\group,k}$ and $D_{\group,k}$ to $K$ coincide.
\end{proof}

Note that for every solution $h$ of $\mathcal I$, the map $\chi_h\colon \binom{\V}{k}\to\mathcal O_{\group,k}^{\V}$ defined by $K\mapsto \{\alpha h|_K\mid \alpha\in\group\}$ defines a solution to $\instance_{\group,k}$.
The next lemma proves that every solution to $\instance_{\group,k}$ is of the form $\chi_h$ for some solution $h$ of $\mathcal I$, provided that $\mathcal I$ is $(k,\ell)$-minimal  and that $\group=\Aut(\rel B)$ for some $k$-homogeneous $\ell$-bounded structure $\rel B$.
\begin{restatable}{lemma}{liftingsolution}\label{lem:lifting-solution}
	Let $1\leq k<\ell$ if $k>1$, and if $k=1$, suppose that $\ell\geq 1$.
	Let $\rel B$ be a $k$-homogeneous $\ell$-bounded structure, let $\rel A$ be a first-order reduct of $\rel B$, and let $\mathcal I$ be a $(k,\ell)$-minimal instance of $\Csp(\rel A)$.
	Then every solution to $\instance_{\Aut(\rel B),k}$ lifts to a solution of $\mathcal I$, i.e., it is of the form $\chi_h$ for some solution $h$ of $\mathcal I$.
\end{restatable}

\begin{proof}
	Let $h\colon \binom{\V}{k}\to\mathcal O_{\group,k}^{\V}$ be a solution to $\instance_{\Aut(\rel B),k}$.
	Recall that $h(K)$ is a $K$-orbit  for any $K\in\binom{\V}{k}$, and one can therefore restrict $h(K)$ to any  $L\subseteq K$ by setting  $h(K)|_L:=\{f|_L\mid f\in h(K)\}$.
	Note that since $\mathcal I$ is $k$-minimal, we have $h(K)|_{K\cap K'}=h(K')|_{K\cap K'}$ for all  $K,K'\in\binom{\V}{k}$.
	
	We now define an equivalence relation $\sim$ on $\V$.
	Suppose first that $k=1$.
	Then every orbit of elements of $\rel B$ under the action of $\Aut(\rel B)$ must be a singleton 
    (for any orbit with two distinct elements $a,b$, the pairs $(a,a)$ and $(a,b)$ would not be in the same orbit but their subtuples of length one would be, so that $\rel B$ would not be $1$-homogeneous).
	In that case, we identify $\mathcal O_{\group,k}^{\V}$ with the domain $B$ itself, and  set $x\sim y$ if and only if  $h(\{x\})=h(\{y\})$; that is,  $\sim$ is essentially the kernel of $h$.
	
	Suppose next that $k\geq 2$, and set $x\sim y$ if and only if there is $K\in\binom{\V}{k}$ containing $x,y$ such that $h(K)|_{\{x,y\}}$ consists of constant maps.
	One could equivalently ask that this holds for \emph{all} $K$ containing $x,y$ by $2$-minimality, and it then follows that this is indeed an equivalence relation by $(2,3)$-minimality of $\mathcal I$.
	Moreover, $h$ descends to $\binom{\V/{\sim}}{k}$: if $K'=\{[v_1]_{\sim},\dots,[v_k]_{\sim}\}$ is a $k$-element set, define $\tilde h(K'):=h(\{v_1,\dots,v_k\})$.
	The definition of $\tilde h$ does not depend on the choice of representatives, by the very definition of $\sim$.
	
	Define a finite structure $\rel C$ with domain $\V/{\sim}$ in the signature of $\rel B$ as follows.	
    Let $\V/{\sim}=\{[v_1]_{\sim},\ldots,[v_n]_{\sim}\}$. We define $\rel C$ such that the relations holding on the tuple $([v_1]_{\sim},\ldots,[v_n]_{\sim})$ in $\rel C$ are the same as the relations holding on an arbitrary tuple $(b_1,\ldots,b_n)\in B^n$ that satisfies the following. For every $m\leq k$ and for all $[v_{i_1}]_{\sim}, \ldots, [v_{i_{m}}]_{\sim}$ pairwise different, the atomic types of $(b_{i_1},\ldots,b_{i_m})$ and $(g([v_{i_1}]_{\sim}),\dots,g([v_{i_m}]_{\sim}))$ agree for all $g\in \tilde h(K)$ and $K\supseteq \{v_{i_1},\ldots,v_{i_m}\}$.
    This construction does not depend on the choice of the tuple $(b_1,\ldots,b_n)$ by the $k$-homogeneity of $\rel B$ and is well-defined by the consistency of the assignment given by $h$.
	
	Finally, note that all substructures of $\rel C$ of size at most $\ell$ embed into $\rel B$.
	Indeed, let $m\leq \ell$ and let $\rel L$ be an $m$-element substructure of $\rel C$, and let $L'\subseteq V$ be an $m$-element set containing one representative for each element of $\rel L$.
	By $(k,\ell)$-minimality of $\mathcal I$, there exists $C\subseteq A^{L'}$ in $\mathcal I$, and a corresponding constraint $C_{\Aut(\rel B),k}$ of $\instance_{\Aut(\rel B),k}$.
	Thus, $h|_{\binom{L'}{k}}\in C_{\Aut(\rel B),k}$, so that there exists $g\in C$ such that for all $K\in\binom{L'}{k}$, $g|_K\in h(K)$.
	Thus $g$ corresponds to an embedding of every $k$-element substructure of $\rel L$ into $\rel B$, and since $\rel B$ is $k$-homogeneous, $g$ is an embedding of $\rel L$ into $\rel B$.
	Finally, since $\rel B$ is $\ell$-bounded, it follows that there exists an embedding $e$ of $\rel C$ into $\rel B$.
	
	It remains to check that the composition of $e$ with the canonical projection $\V\to\ \V/{\sim}$ is a solution to $\mathcal I$, which is trivial since the relations of $\rel A$ are unions of orbits under $\Aut(\rel B)$.
\end{proof}

For every finite set $\V$ and for every $K\in\binom{\V}{k}$, every operation $f$ that is canonical with respect to a permutation group $\group$ induces an operation on the set of orbits of $K$-tuples under $\group$.
We denote this operation by $f_{\group}^{K}$. Finally, we denote by $f_{\group,k}^{\V}$ the union of $f_{\group}^{K}$ for all $K\in\binom{\V}{k}$ and we call it a \emph{multisorted} operation, i.e., this operation is defined only on tuples where all elements belong to the same $\mathcal{O}_{\group}^{K}$ for some fixed $K\in\binom{\V}{k}$. We say that a multisorted operation $f_{\group,k}^{\V}$ is a \emph{multisorted WNU} if $f_{\group}^{K}$ satisfies the WNU identities for every $K\in \binom{\V}{k}$. Similarly, we define multisorted totally symmetric operations. We remark that we could have avoided the use of multisorted functions by using $k$-tuples of elements of $\V$ instead of $k$-element subsets of $\V$ in Definition~\ref{def:instance-fin}; however, encoding the original instance this way  would have added considerable  redundancy and made the proof more technical.

Note that in the above situation, for every constraint $C$ of an instance $\instance=(\V,\mathcal{C})$  it makes formally sense to ask whether $f_{\group,k}^{\V}$ preserves $C_{\group,k}$. Indeed, let $f$ be $n$-ary. Whenever $K$ is a $k$-element subset of the scope of $C$, and hence a variable in the scope of $C_{\group,k}$, then for all $g_1,\ldots,g_n\in C_{\group,k}$ we have $g_1(K),\ldots,g_n(K)\in \mathcal{O}_{\group}^{K}$; hence $f_{\group,k}^{\V}$ can be applied to these values, and doing so for all variables in the scope of $C_{\group,k}$ altogether yields a function from this scope to $\mathcal O_{\group,k}^{\V}$ which can be an element of $C_{\group,k}$ or not.

\begin{restatable}{lemma}{compatiblecanonical}\label{lem:compatible-canonical-polymorphisms}
    Let $f$ be a polymorphism of $\rel A$ that is canonical with respect to $\group$.
    Every constraint in 
    $\instance_{\group,k}$ is preserved by $f_{\group,k}^{\V}$.
\end{restatable}

\begin{proof}
    Let $n$ be the arity of $f$ and let $C\subseteq A^U$ be a constraint in $\mathcal I$.
    In particular, since $\mathcal I$ is an instance of $\Csp(\rel A)$, and $f$ is a polymorphism of $\rel A$, we have that $C$ is preserved by $f$.
    
    Let $C_{\group,k}$ be the constraint in $\instance_{\group,k}$ given by $C$, and let $g_1,\dots,g_n\in C_{\group,k}$.
    By definition, for every $i\in\{1,\dots,n\}$ there is $g'_i\in C$ such that for all $K\in\binom{U}{k}$, $g'_i|_K\in g_i(K)$.
    Note that $f(g'_1|_K,\dots,g'_n|_K)=f(g'_1,\dots,g'_n)|_K$,
    so that $f(g'_1,\dots,g'_n)|_K \in f^{\V}_{\group,k}(g_1(K),\dots,g_n(K))$.
    Since $f(g'_1,\dots,g'_n)\in C$, it follows that $f^{\V}_{\group,k}(g_1,\dots,g_n)$ is in $C_{\group,k}$.
\end{proof}

The following corollary of Theorem~\ref{thm:characterization-bwidth} and Lemma~\ref{lem:compatible-canonical-polymorphisms} shows that if $\rel A$ has as polymorphisms canonical multisorted WNU operations of all arities $m\geq 3$, and if $\instance_{\group,k}$ is non-trivial and $(2,3)$-minimal, then it has a a solution.

\begin{corollary}\label{cor:multisorted_WNUs}
Suppose that $|\V|\geq k$, and that $\rel A$ has
for every $m\geq 3$ a polymorphism $f_m$ that is canonical with respect to $\group$ and such that $(f_m)_{\group,k}^{\V}$ is a multisorted WNU operation. If $\instance_{\group,k}$ is non-trivial and $(2,3)$-minimal, then it has a solution.
\end{corollary}

\begin{proof}
    Let $\rel D$ be the finite structure with domain $\mathcal O_{\group,k}^{\V}$ which contains for every constraint $C_{\group,k}$ in $\instance_{\group,k}$ the corresponding relation. More precisely, for every constraint $C_{\group,k}$ with $C\subseteq A^U$ for some $U\subseteq \V$, $\rel D$ contains the relation $R_C:=\{(g(K))_{K\in\binom{U}{k}}\mid g\in C_{\group,k}\}$.

    Let $m\geq 3$ be arbitrary. We will now prove that $\Pol(\rel D)$ contains a WNU operation of arity $m$. Since $|\V|\geq k$, there exists $K'\in\binom{\V}{k}$; let us fix such $K'$, and let $O\in\mathcal O^K_{\group}$ be arbitrary. Let $f:=(f_m)_{\group,k}^{\V}$, and let us define an $m$-ary function $f'\colon (\mathcal O_{\group,k}^{\V})^m\rightarrow \mathcal O_{\group,k}^{\V}$ as follows. We define $f'$ to be equal to $f$ everywhere, where $f$ is defined, and we define $f'(O_1,\dots,O_m):=O$ for every $(O_1,\dots,O_m)\in (\mathcal O_{\group,k}^{\V})^m$, where $f$ is not defined. Note that $f'$ is a WNU operation, and since $f$ preserves all constraints of $\instance_{\group,k}$ by Lemma~\ref{lem:compatible-canonical-polymorphisms}, so does $f'$. It follows that $f'$ preserves all relations $R_C$ of $\rel D$, whence it is a polymorphism of $\rel D$.
    
    Finally, the implication from item (\ref{itm:wnu}) to item (\ref{itm:width-2-3}) in Theorem~\ref{thm:characterization-bwidth} yields that $\rel D$ has relational width $(2,3)$. Since $\instance_{\group,k}$ is an instance of $\Csp(\rel D)$ by definition, and since it is non-trivial and $(2,3)$-minimal by the assumption, $\instance_{\group,k}$ has a solution.
\end{proof}

In order to prove the second item of Theorem~\ref{thm:canonicalbwidth}, we need also the following version of Corollary~\ref{cor:multisorted_WNUs}.

\begin{corollary}\label{cor:multisorted_symmetric}
Suppose that $|\V|\geq k$, and that $\rel A$ has for every $m\geq 1$ a polymorphism $f_m$ that is canonical with respect to $\group$ and such that $(f_m)_{\group,k}^{\V}$ is a multisorted totally symmetric operation. If $\instance_{\group,k}$ is non-trivial and $(1,1)$-minimal, then it has a solution.
\end{corollary}

\begin{proof}
    The proof is identical to the proof of Corollary~\ref{cor:multisorted_WNUs}, using Theorem~\ref{thm:ts-finite} instead of Theorem~\ref{thm:characterization-bwidth}.
\end{proof}

Finally, this allows us to prove Theorem~\ref{thm:canonicalbwidth} from the introduction.

\canonicalbwidth*

\begin{proof}
    Suppose that the assumption of the first item of Theorem~\ref{thm:canonicalbwidth} is satisfied. 
    Let $\instance$ be a non-trivial $(2k,\max(3k,\ell))$-minimal instance of $\Csp(\rel A)$, and let $\instance_{\Aut(\rel B),k}$ be the associated instance from Definition~\ref{def:instance-fin}. Note that $\instance_{\Aut(\rel B),k}$ is a $\Csp$ instance over a finite set by the $\omega$-categoricity of $\rel B$. Note moreover that any $(2k,\max(3k,\ell))$-minimal non-trivial instance with less than $k$ variables admits a solution. Hence, we may assume that $\instance$ has at least $k$ variables.
    By Lemma~\ref{lem:ak-bk-minimal}, $\instance_{\Aut(\rel B),k}$ is a $(2,3)$-minimal instance, and it is non-trivial by definition and since $\instance$ has at least $k$ variables. 
    Corollary~\ref{cor:multisorted_WNUs} yields that $\instance_{\Aut(\rel B),k}$ admits a solution. Since $\mathcal I$ is $(2k,\max(3k,\ell))$-minimal, it is also  $(k,\ell)$-minimal, and hence this solution lifts to a solution of $\mathcal I$ by Lemma~\ref{lem:lifting-solution}.
    Thus, $\rel A$ has relational width at most $(2k,\max(3k,\ell))$.
    
    Suppose now that the assumption in the second item is satisfied. If $k=1$, then the orbits of elements under $\Aut(\rel B)$ are singletons, whence $\rel B$ is a finite structure and the result follows immediately from Theorem~\ref{thm:ts-finite}. Otherwise, by the same reasoning as above but using Corollary~\ref{cor:multisorted_symmetric} instead of Corollary~\ref{cor:multisorted_WNUs}, given a $(k,\max(k,\ell))$-minimal instance $\instance$, the associated instance $\instance_{\Aut(\rel B),k}$ is $(1,1)$-minimal and therefore has a solution.
    Since $\instance$ is $(k,\max(k,\ell))$-minimal, and if $k=2$, then it is also $(2,3)$-minimal, this solution lifts to a solution of $\instance$ by Lemma~\ref{lem:lifting-solution}. 
\end{proof}

\section{A New Loop Lemma for Smooth Approximations}\label{sect:looplemma}

We refine the algebraic theory of smooth approximations from~\cite{SmoothApproximations}. Building on deep algebraic results from~\cite{BartoKozikCyclic} on   finite  idempotent algebras  that are equationally non-trivial, we lift some of the theory from binary symmetric  relations to cyclic relations of arbitrary arity. 

\subsection{The loop lemma}

\begin{definition}
Let $R$ be a relation of arity $m\geq 1$, let $1\leq n\leq m$,  and let $i_1,\ldots,i_n$ be pairwise distinct 
elements from $\{1,\ldots,m\}$. We denote by $\proj{(i_1,\ldots,i_n)}{R}$ the projection of $R$ to the coordinates $(i_1,\ldots,i_n)$, i.e., $\proj{(i_1,\ldots,i_n)}{R}=\{(b_{i_1},\ldots,b_{i_n})\mid (b_1,\ldots,b_{m})\in R\}$.

The \emph{linkedness congruence} 
of a binary  relation $R\subseteq A\times B$ is the equivalence relation $\lambda_R$ on $\proj{(2)}{R}$  defined by $(b,b')\in \lambda_R$ if there are $k\geq 0$ and  $a_0,\dots,a_{k-1}\in A$ and $b=b_0,\dots,b_k=b'\in B$ such that $(a_i,b_i)\in R$ and $(a_i,b_{i+1})\in R$ for all $i\in\{0,\dots,k-1\}$. We say that $R$ is \emph{linked} if it is non-empty and  $\lambda_R$ relates any two elements of $\proj{(2)}{R}$. If $R$ is of arity greater than two, then we define linkedness as before   viewing $R$ as a binary relation   between $\proj{(1,\ldots,m-1)}{R}$ and $\proj{(m)}{R}$.

If $A$ is a set and $m\geq 2$, then we call a relation $R\subseteq A^m$  \emph{cyclic} if it is invariant under cyclic permutations of the components of its tuples, i.e., for every cyclic permutation $\sigma$ of $(1,\dots,m)$, and for every $(a_1,\dots,a_m)\in R$, it holds that $(a_{\sigma(1)},\dots,a_{\sigma(m)})\in R$. The \emph{support} of a relation $R$ is the set of all elements of $A$ which appear in some tuple of $R$. Note that if $R$ is cyclic, then its support equals  the projection of $R$ to any of its arguments.
\end{definition}

If $R$ is invariant under an oligomorphic group action on $A\times B$, 
then there is an upper bound on the length $k$  to witness  $(b,b')\in\lambda_R$, and therefore $\lambda_R$ is pp-definable from $R$; in particular, it is invariant under any function clone acting on $A\times B$ and preserving $R$. 
\begin{definition}
Let $\gG$ be a permutation group on a set $A$. A \emph{pseudo-loop with respect to $\gG$} is a tuple of elements of $A$ all of whose components belong to the same $\gG$-orbit~\cite{Pseudo-loop,BartoPinskerDichotomy,Topo}. If $\gG$ contains only the identity function, then a pseudo-loop is called a \emph{loop}.
\end{definition}

Our next goal is to prove that any cyclic linked relation on a finite domain which is invariant under an equationally non-trivial idempotent function clone contains a loop. This result will be used to obtain a generalization of the second loop lemma of smooth approximations~\cite[Theorem~11]{SmoothApproximations}. In order to refer to  results on finite algebras more easily, we will use the language of algebras rather than clones in the proof. We need the following definitions.

\begin{definition}
Let $\aA$ be an algebra and let $B$ be a subuniverse of $\aA$, i.e., a subset of the universe of $\aA$ closed under all operations of $\aA$. We say that $B$ is an \emph{absorbing subuniverse} of $\aA$ (or $B$ absorbs $\aA$) if there exists a term operation $f$ of arity $n\geq 2$ such that for any $j\in\{1,\ldots,n\}$ and for any $(a_1,\ldots,a_n)\in A^n$ with $a_i\in B$ for all $i\neq j$, $f(a_1,\ldots,a_n)\in B$.

If $B$ is an absorbing subuniverse of $\aA$ and no proper subuniverse of $B$ absorbs $\aA$, we call $B$ a \emph{minimal absorbing subuniverse} of $\aA$ and write $B\minabs \aA$.
\end{definition}

We will use the following results from~\cite{BartoKozikCyclic}.

\begin{proposition}[Proposition 2.15 in \cite{BartoKozikCyclic}]\label{prop:absorbing}
Let $\aA_1,\aA_2$ be algebras on finite domains whose polymorphism clones are equationally non-trivial, and let $\aR$ be a subalgebra of $\aA_1 \times \aA_2$ with $\proj{(1)}{R}=A_1,\proj{(2)}{R}=A_2$. Then both of the following hold.
\begin{enumerate}
    \item If $R$ is linked, $\aB_1\minabs \aA_1$, $\aB_2\minabs \aA_2$, and $(B_1\times B_2)\cap R\neq \emptyset$, then $\aB_1\times \aB_2\minabs R$.
    \item If $R$ is linked, $\aB_1\minabs \aA_1$ or $\aB_1\minabs \aA_2$, $\aB_2\minabs \aA_1$ or $\aB_2\minabs \aA_2$, $b_1\in\aB_1$, and $b_2\in\aB_2$, then $b_1$ and $b_2$ can be linked via $a_0,\dots,a_k$ for some $k\geq 1$, where $a_i$ is contained in a minimal absorbing subalgebra of $\aA_1$ or $\aA_2$ for every $i\in\{0,\dots,k\}$.
\end{enumerate}
\end{proposition}

We say that a digraph has \emph{algebraic length 1} if it contains a path from some element to itself  whose numbers of forward arcs and backward arcs differ by 1.

\begin{theorem}[Theorem 3.6 in \cite{BartoKozikCyclic}]\label{thm:loop}
    Let $\aA$ be an algebra on finite domain with equationally non-trivial polymorphism clone, and let $\rel G = (A; H)$ be a digraph of algebraic length 1 such that every vertex has an outgoing and an incoming edge and such that $H$ is a subuniverse of $\aA^2$. Then $\rel G$ contains a loop. Moreover, if there exists an absorbing subuniverse $I$ of $\aA$ which is contained in a weak component of $\rel G$ of algebraic length 1, then the loop can be found in a minimal absorbing subuniverse of $\aA$.
\end{theorem}

\begin{restatable}[Consequence of the proof of Theorem~4.2 in \cite{BartoKozikCyclic}]{theorem}{cyclicterm}\label{thm:cyclic-term}
    Let $\cC$ be an idempotent function clone on a finite domain  that is  equationally non-trivial.
    Then any 
     $\cC$-invariant cyclic linked  relation on its domain contains a loop.
\end{restatable}

\begin{proof}
Let $R$ as in the statement be given, and denote its arity by $m$; we may assume $m\geq 2$. Given $1\leq i\leq m$, we set $R_i:=\proj{(1,\ldots,i)}{R}$; moreover, we set $R_{(i,j)}:=\proj{(i,j)}{R}$ for all distinct $i,j$ with $1\leq i,j\leq m$.  

We denote the support of $R$ by  $A$. Note that for all $i\in\{1,\dots,m-1\}$ we have that $R_{i+1}$ is linked  when viewed as a binary relation between $R_i$ and $A$. We give the short argument showing  linkedness for the convenience of the reader. Let $a,b\in A$ be arbitrary; since they are linked in $R_m$, there exist  $k\geq 1$ and  $a=c_0,\ldots,c_{2k}=b$ such that $(c_{2j+1},c_{2j})\in R_m$ and $(c_{2j+1},c_{2j+2})\in R_m$
for all $0\leq j<k$. For all such $j$, we have $c_{2j+1}\in R_{m-1}$; moreover,  $\proj{(\ell,\ldots,m-1)}{c_{2j+1}}\in R_{m-\ell}$ for all $1\leq \ell\leq m-1$ by the  cyclicity of $R$. Hence, these elements prove   linkedness of $a,b$ in $R_{i+1}$ for all $1\leq i\leq m-1$.

As a consequence, 
$R_{(1,i)}$ is linked for all $i\in\{2,\dots,m\}$. Indeed, for arbitrary $a,b\in A$, if  $a=c_0,\ldots,c_{2k}=b$ witness the linkedness of $a,b$ in $R_i$, then  $c_0,\proj{(1)}{c_1},c_2,\dots,\proj{(1)}{c_{2k-1}},c_{2k}$ witness the linkedness of $a,b$ in $R_{(1,i)}$.

Let $\alg A$ be the algebra whose domain is the support $A$ of $R$ and whose fundamental operations are the restrictions of the functions of the clone  $\cC$ to that domain, i.e., $\alg A$ is the algebra $(A;\{f|_A\mid f\in \cC\}$); note that this is well-defined since $A$, as the projection of $R$ to any coordinate, is invariant under the functions of $\cC$. Let $\alg R$ be the subalgebra of $\alg A^m$ with domain $R$. Similarly, we write $\aR_i$ for the subalgebra of $\aA^i$ on the domain $R_i$, for all $1\leq i\leq m$. 
Following the proof of Theorem 4.2 in~\cite{BartoKozikCyclic}, we prove by induction on $i\in\{1,\dots,m\}$ that the following properties hold:
\begin{enumerate}
    \item There exists $I\minabs\alg A$ such that $ I^i\minabs \alg R_i$.
    \item For all $ I_1,\dots, I_i\minabs\alg A$ such that $R_i\cap (I_1\times\cdots\times I_i)\neq \emptyset$,
    we have $ R_i\cap ( I_1\times\cdots\times  I_i)\minabs \alg R_i$.
\end{enumerate}
    
The case $i=1$ is trivial, since $R_1=A$. 
The case $i=m$ gives us a constant tuple in $R$.

We prove property (2) for $i+1$.
Let $ I_1,\dots, I_{i+1}\minabs\alg A$ be such that $R_{i+1}\cap(I_1\times\cdots\times I_{i+1})\neq\emptyset$.
By the induction hypothesis, $ R_i\cap( I_1\times\cdots\times  I_i)\minabs \alg R_i$.
By the argument above, $R_{i+1}$ is linked as a relation between  $R_i$ and $A$.
Thus, the assumptions of item (1) of Proposition~\ref{prop:absorbing} are satisfied, and we have $ R_{i+1}\cap ( I_1\times\cdots\times  I_{i+1})\minabs \alg R_{i+1}$.
    
We prove property (1) for $i+1$.
Define a directed graph $H$ on $R_i$ by setting 
\[H:= \{((a_1,\dots,a_i),(a_2,\dots,a_{i+1}))\mid (a_1,\dots,a_{i+1})\in R_{i+1}\}.\]
Let $ I\minabs\alg A$ be such that $I^i\minabs \alg R_i$, which exists by the induction hypothesis.
We show that:
\begin{itemize}
    \item $I^i$ is a subset of a (weak)  connected component of $H$; 
    \item this connected component has algebraic length 1.
\end{itemize}
For the first item, let $X:=\{ x \mid \exists a_1,\ldots,a_i\in I\; ( (a_1,\dots,a_i,x)\in R_{i+1})\}$, which is an absorbing subuniverse of $\alg A$.
Let $ X_1\subseteq X$ be a minimal absorbing subuniverse of $\aA$. Then since  $(I^i\times X_1)\cap R_{i+1}\neq \emptyset$ property~(2) gives us  $(I^i\times X_1)\subseteq R_{i+1}$. Reiterating this idea, we find minimal absorbing subuniverses $X_2,\ldots, X_i$ of $\aA$ such that for all $1\leq j\leq i$ we have that  $I^{i-j+1}\times X_1\times\cdots\times X_{j}$ is contained in $R_{i+1}$. Now pick an arbitrary tuple $(a_1,\ldots,a_i)\in I^i$, and an arbitrary tuple $(x_1,\ldots,x_i)\in X_1\times\cdots\times  X_i$. Then there is a path in $H$ from $(a_1,\ldots,a_i)$ to $(x_1,\ldots,x_i)$ by the above, proving the first item.
    
For the second item, let
\[E:=\{(x,y) \mid \exists c_2,\dots,c_i\in I\; ( (x,c_2,\dots,c_i,y)\in R_{i+1})\}.\]
Let $V_1,V_2$ be the projection  of $E$ onto its  first and second coordinate, respectively; then $E$ is  a relation between $V_1$ and $V_2$.
We have that $E$ is an absorbing subuniverse of the algebra  $\alg R_{(1,i+1)}$ induced by $R_{(1,i+1)}$ in $\aA^2$. By assumption on $R$, we have that $R_{(1,i+1)}$ is linked. 
Therefore, $E$ is linked, and $V_1$ and $V_2$ are absorbing subuniverses of  $\alg A$.
Note that $I\subseteq V_1$ and $I\subseteq V_2$.
Let $b\in I$ be arbitrary. Then there exist $k\geq 0$ and  $c_0,\ldots,c_{2k+1}$ such that $(c_{2j},c_{2j+1})\in E$ for all $0\leq j\leq k$ and $(c_{2j+2},c_{2j+1})\in E$ for all $0\leq j\leq k-1$ and such that $c_0=b=c_{2k+1}$. These elements can be assumed to lie in minimal absorbing subuniverses of $\aA$ by item (2) of Proposition~\ref{prop:absorbing}. We then have, by property~(2),  $(c_{2j},b,\ldots,b,c_{2j+1})\in R_{i+1}$ for all $0\leq j\leq k$ and $(c_{2j+2},b,\ldots,b,c_{2j+1})\in R_{i+1}$ for all $0\leq j\leq k-1$. This gives a path of algebraic length~1 in $H$ from $(b,\ldots,b)$ to itself. This proves the second item.
    
Since $R$ is linked, it follows that every vertex of $H$ has an outgoing and an incoming edge.
Moreover, $H$ has algebraic length $1$, and is invariant under $\clone$, which  is equationally non-trivial. 
Hence, by Theorem~\ref{thm:loop}, there exists a loop in $H$ which lies in a minimal absorbing subuniverse $K$ of $\alg R_i$. By the definition of $H$, this loop is a constant tuple $(a,\ldots,a)$.
By projecting $K$ on the first component, we obtain a minimal absorbing subuniverse $J \minabs\alg A$; since $a\in J$, we have that
$J^{i+1}\cap R_{i+1}\neq\emptyset$. By property~(2), we get that $J^{i+1}\minabs \alg R_{i+1}$, so that property~(1) holds.
\end{proof}

The following is the previously announced  generalization of~\cite[Theorem~11]{SmoothApproximations} from binary symmetric relations to arbitrary cyclic relations. 
\begin{proposition}\label{prop:approximation-dichotomy}
    Let $n\geq 1$, and let $\cD$ be an oligomorphic function clone  on a set $A$ which is a model-complete core. Let $\cC\subseteq \cann{\cD}$ be such that  $\cC^n/{\gG_\cD}$ is equationally non-trivial. Let $(S,\sim)$ be a  minimal subfactor of the action  $\cC^n$ with $\gG_\cD$-invariant $\sim$-classes. 
    Then for every $\cD$-invariant cyclic relation $R$ with support $\langle S\rangle_\cD$
    one of the following holds:
    \begin{enumerate}
        \item The linkedness congruence of $R$ is a $\cD$-invariant  approximation  of $\sim$.
        \item $R$ contains a pseudo-loop with respect to  $\gG_\cD$.
    \end{enumerate}
\end{proposition}
\begin{proof}
Let $R$ be given, and denote  its arity by $m$. Assuming that~(1) does not hold, we prove~(2).

Denote by $\mathcal O$ the set of orbits of $n$-tuples under the action of  $\gG_\cD$ thereon.
	Let $R'$ be the relation obtained by considering $R$ as a relation on $\mathcal O$, i.e., 
	\[ R' := \{(O_1,\dots,O_m)\in\mathcal O^m \mid R\cap (O_1 \times \dots \times O_m) \neq\emptyset\}.\]
	Thus, $R'$ is an $m$-ary cyclic relation with support  $S'\subseteq\mathcal O$,
	and $R'$ contains a loop  if and only if  $R$ satisfies~(2).
	
	By assumption, the action $\cC^n/\gG_\cD$ is equationally non-trivial; moreover, it is idempotent since $\cD$ is a model-complete core. Note also that $R'$, and in particular  $S'$, are  preserved by this action. It is therefore sufficient to  show that $R'$ is linked and apply  Theorem~\ref{thm:cyclic-term}.
	
	Recall that we  consider $R$ also as a binary relation between $\proj{m-1}{R}$ and $\cl{S}$; similarly, we consider $R'$ as a binary relation between $\proj{m-1}{R'}$ and $S'$. By the  oligomorphicity  of $\cD$, the linkedness congruence $\lambda_{R}$ of $R$ is  invariant under  $\cD$. 

	By our assumption that~(1) does not hold, there exist $c,d\in S$ which are not $\sim$-equivalent and such that $\lambda_{R}(c,d)$ holds; otherwise, $\lambda_{R}$ would be an  approximation  of $\sim$. This implies that the orbits $O_c,O_d$ of $c,d$ are related via $\lambda_{R'}$. 
	By  the minimality of $(S,\sim)$,  we have that $\cl{S}=\langle \{c,d\}\rangle_\cD$. Since $\cD$ is a model-complete core, it preserves the $\gG_\cD$-orbits, and it follows that any tuple in $\cl{S}= \langle \{c,d\}\rangle_\cD$ is $\lambda_R$-related to a tuple in the orbit of $c$. Hence,  $\lambda_{R'}=(S')^2$, and thus $R'$ is linked. Theorem~\ref{thm:cyclic-term} therefore  implies that $R'$ contains a loop, and hence $R$ contains a pseudo-loop with respect to $\gG_\cD$, which is what we had to show.
\end{proof}

The following is a generalization of Lemma~14 in~\cite{SmoothApproximations} from binary relations and functions to relations and functions of higher arity. It is used in Section~\ref{sect:newboundedwidth} to obtain characterization of bounded width for first-order reducts of unary structures and of structures related to the logic MMSNP.

\begin{restatable}{lemma}{narycyclic}\label{lem:n-ary-weird-cyclic}
Let $\cD$ be a clone of polymorphisms of an $\omega$-categorical relational structure on a 
 set $A$, and assume that $\cD$
is a model-complete core. Let $n\geq 1$, and let $\sim$ be an equivalence relation on a set $S\subseteq A^n$ with  $\gG_\cD$-invariant classes.
    Let $m\geq 1$, and let $P$ be an $m$-ary relation on $\cl{S}$. 
    Suppose that every $m$-ary  $\cD$-invariant  cyclic relation $R$ on $\langle S\rangle_\cD$ which contains a tuple in $P$ with components in at least two   $\sim$-classes  contains a pseudo-loop with respect to  $\gG_\cD$.
    
    Then  there exists an $m$-ary $f\in\cD$ such that for all $a_1,\dots,a_m\in A^n$ we have that if 
	the tuple  $(f(a_1,\dots,a_m),f(a_2,\dots,a_m,a_1),\dots,f(a_m,a_1,\dots,a_{m-1}))$ is in $P$, then it intersects at most one   $\sim$-class. 
\end{restatable}

\begin{proof}
	The proof is similar to the proof of Lemma~14  in~\cite{SmoothApproximations}. 
    Fix $m\geq 2$. 
	Call an $m$-tuple $(a_1,\dots,a_m)$ of elements of $A^n$ \emph{troublesome} if there exists an  $m$-ary  $r\in\cD$ such that
	\[(b_1,\dots,b_m):=(r(a_1,\dots,a_m),r(a_2,\dots,a_m,a_1),\dots,r(a_m,a_1,\dots,a_{m-1}))\]
	is in $P$ and has components in at least two  $\sim$-classes.
    We first show that for any tuple $(a_1,\dots,a_m)\in (A^n)^m$, there exists an $m$-ary function $g$ in $\cD$ such that $$(g(a_1,\ldots,a_m),\ldots,g(a_m,a_1,\ldots,a_{m-1}))$$ is not troublesome.
    We then use this fact to obtain a sequence $(f^i)_{i\in\mathbb N}$ of operations from $\cD$ satisfying the same property on subsets of $A^n$ of increasing size, and finally obtain a function satisfying the conclusion of the lemma by a standard compactness argument.
    
    Note that if $(a_1,\ldots,a_m)$ is not troublesome and $h\in\cD$, then $(h(a_1,\ldots,a_m),\ldots,h(a_m,a_1,\ldots,a_{m-1}))$ is not troublesome either.
    Moreover, if $a_1,\ldots,a_m$  are in the same orbit under $\gG_\cD$, then $(a_1,\ldots,a_m)$ is not troublesome: since $\cD$ is a model-complete core, $h(a_1,\ldots,a_m),\ldots,h(a_m,a_1,\ldots,a_{m-1})$ are also in the same orbit under $\gG_\cD$ for any $m$-ary $h\in\cD$; since the classes of $\sim$ are closed under $\gG_\cD$, our claim follows. 
	
	For each troublesome tuple $(a_1,\ldots,a_m)$, the smallest $\cD$-invariant relation containing the  set \[\{(b_1,\dots,b_m),\dots,(b_m,b_1,\dots,b_{m-1})\}, \] where $(b_1,\ldots,b_m)$ is defined via a witnessing function $r$ as above, is cyclic, contains a tuple in $P$ with components in at least two $\sim$-classes,  and its support is contained in  $\cl{S}$.  Hence, it contains a pseudo-loop with respect to $\gG_\cD$ by our assumptions. 
	This implies that there exists an $m$-ary $g\in\cD$ such that the entries of the tuple  \[(g(b_1,\dots,b_m),\dots,g(b_m,b_1,\dots,b_{m-1}))\] all belong to the same  $\gG_\cD$-orbit. The function \[s(x_1,\ldots,x_m):=g(r(x_1,\ldots,x_m),\ldots,r(x_m,x_1,\ldots,x_{m-1}))\] 
	thus has the property that the entries of 
	\[(s(a_1,\dots,a_m),\dots,s(a_m,a_1,\dots,a_{m-1}))\] all lie in the same $\gG_\cD$-orbit.
	
	In conclusion, for every tuple $(a_1,\ldots,a_m)$ of elements of $A^n$ -- troublesome or not -- there exists an $m$-ary function $d$ in $\cD$ such that $(d(a_1,\ldots,a_m),\ldots,d(a_m,a_1,\ldots,a_{m-1}))$ is not troublesome: if $(a_1,\ldots,a_m)$ is troublesome one can take $d$ to be the operation $s$ we just described; if $(a_1,\ldots,a_m)$ is not troublesome then the first projection works.
	
	Let $((a_1^i,\ldots,a_m^i))_{i\in\mathbb N}$ be an enumeration of all  $m$-tuples of elements of  $A^n$.
    We build by induction on $i\in\mathbb N$ an operation $f^i\in\cD$ such that  $$(f^i(a^j_1,\ldots,a^j_m),\ldots,f^i(a_m^j,a_1^j,\ldots,a^j_{m-1}))$$ is not troublesome for any $j<i$.
    For $i=0$ there is nothing to show, so suppose that $f^i$ is built.
    Let $d\in\cD$ be an $m$-ary operation making the tuple 
    $$(f^i(a^i_1,\ldots,a^i_m),\ldots,f^i(a_m^i,a_1^i,\ldots,a^i_{m-1}))$$
    not troublesome, in the manner of  the previous paragraph.
    Then setting $$
    f^{i+1}(x_1,\ldots,x_m):=d(f^i(x_1,\ldots,x_m),\ldots,f^i(x_m,x_1,\ldots,x_{m-1}))
    $$ clearly has the desired property for the tuple $(a_1^i,\ldots,a_m^i)$, and 
    also for all tuples $(a_1^j,\ldots,a_m^j)$ with $j<i$ by the remark in the first paragraph.
    
    By a standard compactness argument, the sequence $(f^i)_{i\in\mathbb N}$ yields the desired function $f$. This follows directly from Lemma 4 in~\cite{BodPin-CanonicalFunctions}. We provide a different argument for the convenience of the reader. More precisely, let us construct the following infinite, finitely branching tree. For $i>0$, let $A_i$ be the union of the entries of all tuples $a_1^j,\ldots,a_m^j$ with $j\leq i$.
    Define an equivalence relation $\sim_i$ on the $k$-ary functions in $\cD$ by $f\sim_i g$ if there exists $\alpha\in\gG_\cD$ such that $\alpha f|_{A_i}=g|_{A_i}$.
    The tree contains a node for every $i>0$ and for every equivalence class $[f]_{i}$ of some $k$-ary $f\in\cD$ with respect to $\sim_i$ with the property that $(f(a^j_1,\ldots,a^j_m),\ldots,f(a_m^j,a_1^j,\ldots,a^j_{m-1}))$ is not troublesome for every $j\leq i$. 
    A node $[g]_{i+1}$ is a child of $[f]_i$ if $g\sim_i f$.
    Since $\gG_\cD$ is the automorphism group of an $\omega$-categorical structure, every $\sim_i$ has only finitely many equivalence classes, and thus every $[f]_i$ has finite degree.
    Since $\gG_\cD$ is oligomorphic, there are only finitely many such orbits $O^i$ for every $i>0$, whence this tree is finitely branching, and it is clearly infinite.
    The tree contains a path of arbitrary large length, as the sequence $[f^{i}]_1,[f^{i}]_2,\dots,[f^{i}]_i$ is a path of length $i$ for any $i$.
    K\"onig's tree lemma then yields that there exists an infinite branch $([f_i]_i)_{i\in\mathbb N}$.
    Up to replacing $f_i$ by $\alpha_if_i$ for some suitably chosen $\alpha_i\in\gG_\cD$, one can assume without loss of generality that $f_{i+1}|_{A_{i}}=f_i|_{A_i}$ holds for all $i\in\mathbb N$.
    Then $f=\bigcup f_i|_{A_i}$ is well-defined and is in $\cD$.
    Moreover, it satisfies the conclusion of the lemma.
\end{proof}

\section{Applications: Collapses of the bounded width hierarchies for some classes of infinite structures}\label{sect:newboundedwidth}

We now apply the algebraic results of Section~\ref{sect:looplemma} and the theory of smooth approximations to obtain a characterisation of bounded width for CSPs of first-order reducts of unary structures and for CSPs in MMSNP. Moreover, we obtain a collapse of the bounded width hierarchy for such CSPs and for CSPs of many other structures studied in the literature.

\subsection{Unary Structures}\label{sect:unary}

We are going to prove the following characterization of  bounded  width for first-order reducts of unary structures.
\begin{theorem}\label{thm:characterization-unary}
    Let $\rel A$ be a first-order reduct of a unary structure,
    and assume that $\rel A$ is a model-complete core.
    Then one of the following holds:
    \begin{itemize}
        \item The clone of $\Aut(\rel A)$-canonical polymorphisms of $\rel A$ is not equationally affine, or equivalently, it contains pseudo-WNU operations modulo $\Aut(\rel A)$ of all arities $n\geq 3$;
        \item $\Pol(\rel A)$ has a uniformly continuous minion  homomorphism to an affine   clone.
    \end{itemize}
\end{theorem}
In the first case, the stated equivalence follows from Section~\ref{subsect:clones} and $\rel A$ has relational width at most $(4,6)$ by Theorem~\ref{thm:canonicalbwidth}, 
and in the second case it does not have bounded width by Theorem~\ref{thm:characterization-bwidth}. We remark that 
by~\cite[Lemma 6.7]{ReductsUnary}, the model-complete core of a first-order reduct of a unary structure is also a first-order reduct of a unary structure. Thus, Theorem~\ref{thm:characterization-unary} gives a characterization of bounded width for \emph{all} first-order reducts of unary structures.

The two items of Theorem~\ref{thm:characterization-unary} are invariant under expansions of $\rel A$ by a finite number of constants (for the first item, this  follows exactly as the preservation of the pseudo-Siggers identity under adding constants  in~\cite{BartoPinskerDichotomy, Topo}, for the second item, \cite[Theorem 1.8]{wonderland} yields that the existence of a uniformly continuous minion homomorphism from $\Pol(\rel A)$ to an arbitrary clone $\cC$ is equivalent to the existence of such a homomorphism from the polymorphism clone of an expansion of $\rel A$ by a finite number of constants.) Thus, by Proposition~6.8 in~\cite{ReductsUnary}, one can assume in the following that $\rel A$ is a first-order \emph{expansion} of 
$(\rel N;V_1,\dots,V_r)$ (and not only a first-order reduct thereof, i.e., the relations $V_1,\ldots,V_r$ are present in $\rel A$) where $V_1,\dots,V_{r}$ form a partition of $\rel N$ in which every set is either a singleton or infinite.
Such partitions were called \emph{stabilized partitions} in~\cite{ReductsUnary}, and we shall also call the structure $(\rel N;V_1,\dots,V_r)$ a stabilized partition.

We will use the following fact which states that $\Pol(\rel A)$ locally interpolates $\can{\Pol(\rel A)}$.

\begin{lemma}[Proposition~6.5 in~\cite{ReductsUnary}]\label{lem:canonisation-unary}
Let $\rel A$ be a first-order expansion of a stabilized partition  $(\rel N;V_1,\dots,V_r)$.
For every $f\in\Pol(\rel A)$  there exists $g\in\can{\Pol(\rel A)}$ which is locally interpolated by $f$ modulo $\Aut(\rel A)$.
\end{lemma}

The following will allow us to assume, in most proofs, the presence of functions in $\Pol(\rel A)$ which are injective on every set in the stabilized partition. This is the analogue to the efforts to obtain binary injections in~\cite{SmoothApproximations}.

\begin{restatable}[Subset of the proof of Proposition 6.6~\cite{ReductsUnary}]{lemma}{injections}\label{lem:injections}
    Let $\rel A$ be a first-order expansion of a stabilized partition  $(\rel N;V_1,\dots,V_r)$, and assume it is a model-complete core. If\/ $\Pol(\rel A)$ has no continuous  
    clone homomorphism to $\Proj$, then it contains operations of all arities whose restrictions to  $V_i$ are injective for all $1\leq i\leq r$. 
\end{restatable}

\begin{proof}
We show by induction that for all $1\leq j\leq r$, there exists a binary operation in $\Pol(\rel A)$ whose restriction to each of $V_1,\ldots, V_j$ is injective; higher arity functions with the same property are then obtained by nesting the binary operation. For the  base case $j=1$, observe that the disequality relation $\neq$ is preserved on $V_1$ since $\rel A$ is a model-complete core; together with the restriction of $\Pol(\rel A)$ to $V_1$ being equationally non-trivial, we then obtain an  operation which acts as as an essential function on $V_1$. This in turn easily  yields a function that acts as a binary injection on $V_1$  -- see e.g.~\cite{ecsps}. For the induction step, assuming the statement holds for   $1\leq j< r$, we show the same for $j+1$. By the  induction hypothesis, there exist binary functions $f,g\in\Pol(\rel A)$ such that the restriction of $f$ to each of the sets $V_1,\ldots,V_{j}$ is injective, and the restriction of $g$ to $V_{j+1}$ is injective. If  the restriction of $f$ to $V_{j+1}$ depends only on its first or only on its second variable, then it is injective in that variable since disequality is preserved on $V_{j+1}$, and hence either the function $f(g(x,y),f(x,y))$ or the function $f(f(x,y),g(x,y))$ has the desired property. If on the other hand the restriction of $f$ to $V_{j+1}$ depends on both variables, then the same argument as in the base case yields a function which is injective on $V_{j+1}$, and this function is still injective on each of the sets  $V_1,\ldots,V_j$.
\end{proof}

\begin{restatable}{proposition}{fromunary}\label{prop:from-unary}
    Let $\rel A$ be a first-order expansion of a stabilized partition  $(\rel N;V_1,\dots,V_r)$, and assume it is a model-complete core. Suppose that $\Pol(\rel A)$ contains operations of all arities whose restrictions to $V_i$ are injective for all $1\leq i\leq r$. Then the following are equivalent:
    \begin{itemize}\item $\can{\Pol(\rel A)}$ is equationally affine;
    \item  $\can{\Pol(\rel A)}\actson \rel N/\Aut(\rel A)$ is equationally affine.
    \end{itemize}
\end{restatable}

\begin{proof}
    Trivially, if $\can{\Pol(\rel A)}\actson \rel N/\Aut(\rel A)$ is equationally affine, then  so is $\can{\Pol(\rel A)}$.
    
    For the other direction, assume that $\can{\Pol(\rel A)}\actson \rel N/\Aut(\rel A)$ is not equationally affine. We will show that $\can{\Pol(\rel A)}$ contains pseudo-WNU operations of all arities $\geq 3$, whence it is not equationally affine. Since $\can{\Pol(\rel A)}\actson \rel N/\Aut(\rel A)$ is not equationally affine,
    $\can{\Pol(\rel A)}$ contains for every $k\geq 3$ a $k$-ary operation $g_k$ whose action on $\rel N/{\Aut(\rel A)}$
    is a WNU operation by the implication  from~\ref{itm:core-affine} to~\ref{itm:wnu} in Theorem~\ref{thm:characterization-bwidth} and since $\rel N/{\Aut(\rel A)}$ is a finite set.  Fix for all such  $k\geq 3$ an operation $f_k$ of arity $k$ whose restriction to  $V_i$ is injective for all $i\in\{1,\ldots,r\}$, and consider the operation 
    \[ h_k(x_1,\dots,x_k) := f_k(g_k(x_1,\dots,x_k),g_k(x_2,\dots,x_k,x_1),\dots,g_k(x_k,x_1,\dots,x_{k-1})).\]
    Then evaluating $h_{k}(a,b,\dots,b)$ for arbitrary $a,b\in \rel N$,
    all the arguments of $f_k$ belong to the same set $V_i$, by the fact that $g_k$ acts on $\rel N/{\Aut(\rel A)}$ as a weak near-unanimity operation.
    Since $\rel A$ is an expansion of $(\rel N;V_1,\dots,V_r)$, we obtain that $h_k(a,b,\dots,b)$ belongs to $V_i$.
    The same is true for any permutation of the tuple $(a,b,\dots,b)$, so that $h_k$ acts as a WNU operation on $\rel N/{\Aut(\rel A)}$.
    By the injectivity of $f_k$ when restricted to $V_i$, it also follows that $h_k$ acts as a  WNU operation on $\rel N^2/{\Aut(\rel A)}$;  hence, it is a pseudo-WNU operation since it preserves the equivalence of orbits of $2$-tuples under  $\Aut(\rel N;V_1,\dots,V_r)$ and since the stabilized partition is $2$-homogeneous. 
    Taking for every $h_k$ an operation in $\can{\Pol(\rel A)}$ locally interpolated by $h_k$ modulo $\Aut(\rel A)$ which exists by Lemma~\ref{lem:canonisation-unary}, we see that $\can{\Pol(\rel A)}$ contains pseudo-WNU operations of all arities $\geq 3$, and hence it is not equationally affine.
\end{proof}

Before finally proving Theorem~\ref{thm:characterization-unary}, we state the following standard result from~\cite{Valeriote} that allows us to derive, under the assumption that $\can{\Pol(\rel A)}\actson A/{\Aut(\rel A)}$ is equationally affine, the existence of an invariant subset $S\subseteq A$ and an equivalence relation $\sim$ on $S$ with finitely many classes such that $\can{\Pol(\rel A)}\actson S/{\sim}$ is itself an affine clone.

\begin{proposition}[Consequence of Proposition 3.1 from~\cite{Valeriote}]\label{prop:eq-affine_subfactor}
    Let $\cC$ be a an idempotent function clone on a finite domain that is equationally affine. Then there exists a minimal subfactor $(S,\sim)$ of $\cC$ such that $\cC$ acts on the $\sim$-classes as an affine clone.
\end{proposition}

\begin{proof}[Proof of Theorem~\ref{thm:characterization-unary}]
Let $\rel A$ as in Theorem~\ref{thm:characterization-unary} be given; by the remark following that theorem, we may without loss of generality  assume that $\rel A$ is a first-order expansion of a stabilized partition $(\rel N;V_1,\ldots,V_r)$. 
Assume henceforth that $\can{\Pol(\rel A)}$ is  equationally affine; we show that $\Pol(\rel A)$ has a uniformly continuous minion homomorphism to an affine clone.

If $\Pol(\rel A)$ has a continuous clone homomorphism to $\Proj$, then we are done. Assume therefore the contrary; then by Lemma~\ref{lem:injections}, $\Pol(\rel A)$ contains for all $k\geq 2$ a $k$-ary operation whose restriction to $V_i$ is injective for all $1\leq i\leq r$. In particular, Proposition~\ref{prop:from-unary} applies, and thus  $\can{\Pol(\rel A)}\actson \rel N/{\Aut(\rel A)}$ is equationally affine. Proposition~\ref{prop:eq-affine_subfactor} yields that there exists a minimal subfactor $(S,\sim)$ of $\can{\Pol(\rel A)}$ such that $\can{\Pol(\rel A)}$ acts on the $\sim$-classes as an affine clone.

Let $R$ be any $\Pol(\rel A)$-invariant cyclic relation with support $\langle S\rangle_{\Pol(\rel A)}$, containing a  tuple with components in pairwise distinct $\Aut(\rel A)$-orbits and which intersects at least two $\sim$-classes.  
By Proposition~\ref{prop:approximation-dichotomy},   $R$ either gives rise to a $\Pol(\rel A)$-invariant approximation of $\sim$, or it contains a pseudo-loop with respect to $\Aut(\rel A)$.
In the first case, the presence of the tuple required above implies smoothness of the approximation: if $t\in R$ is such a tuple, $c\in \langle S\rangle_{\Pol(\rel A)}$ appears in $t$, and $d\in \langle S\rangle_{\Pol(\rel A)}$ belongs to the same $\Aut(\rel A)$-orbit as $c$, then there exists an element of $\Aut(\rel A)$ which sends $c$ to $d$ and fixes all other elements of $t$.
Hence, $c$ and $d$ are linked in $R$, and the entire $\Aut(\rel A)$-orbit of $c$ is contained in a class of the linkedness relation of $R$.
Thus,  $\Pol(\rel A)$ admits a uniformly continuous minion homomorphism to an affine clone  by Theorem~\ref{thm:fundamental}. 

Hence we may assume that for any $R$ as above the second case holds.  
We now show that this leads to a contradiction, finishing the proof of Theorem~\ref{thm:characterization-unary}. 
We apply Lemma~\ref{lem:n-ary-weird-cyclic} with $n=1$, $m\geq 2$ arbitrary and $P$ the set of  $m$-tuples with entries in pairwise distinct $\Aut(\rel A)$-orbits  within $\langle S\rangle_{\Pol(\rel A)}$, i.e., $P=\{(a_1,\dots,a_m)\in (\langle S\rangle_{\Pol(\rel A)})^m\mid \forall 1\leq i<j\leq m, \forall\alpha \in \Aut(\rel A), \alpha(a_i)\neq a_j\}.$
We obtain an $m$-ary  function $f\in\Pol(\rel A)$ with the property that the tuple \[(f(a_0,\ldots,a_{m-1}),\ldots, f(a_1,\ldots,a_{m-1},a_0))\] intersects at most one $\sim$-class whenever it has entries in pairwise distinct $\Aut(\rel A)$-orbits, 
for all $a_0,\ldots,a_{m-1}\in S$.   
Let $(\rel A,<)$ be the expansion of $\rel A$ by a linear order that is convex with respect to the partition $V_1,\ldots,V_r$ and  dense and without endpoints  on every infinite set of the partition. 
The structure $(\rel A,<)$ 
is a Ramsey structure, since $\Aut(\rel A,<)$ is isomorphic as a  permutation group to the action of the product $\prod^r_{i=1} \Aut(V_i;<)$, and each of the groups of the product is either trivial or the automorphism group of  a Ramsey structure~\cite{Topo-Dynamics}. 
By Theorem~\ref{thm:canonical-homo} we may assume that $f$ is diagonally canonical with respect to $\Aut(\rel A,<)$. Let $a,a'\in A^m$ be so that  $a_i,a_i'$ belong to the same orbit with respect to $\Aut(\rel A)$ for all $1\leq i\leq m$. Then there exists $\alpha\in \Aut(\rel A,<)$ such that $\alpha(a)=a'$, and hence $f(a)$ and $f(a')$ lie in the same $\Aut(\rel A)$-orbit by diagonal canonicity; hence $f$ is $1$-canonical with respect to $\Aut(\rel A)$. Applying Lemma~\ref{lem:canonisation-unary}, we obtain a canonical function $g\in\can{\Pol(\rel A)}$ which acts like $f$ on $\rel N/{\Aut(\rel A)}$.  The property of $f$ stated above then implies for $g$ that  $g(a_0,\ldots,a_{m-1})\sim g(a_1,\ldots,a_{m-1},a_0)$ for all $a_0,\ldots,a_{m-1}\in S$  such that the values  $g(a_0,\ldots, a_{m-1}),\ldots, g(a_{m-1},a_0,\ldots,a_{m-2})$  lie in pairwise distinct $\Aut(\rel A)$-orbits. 

By the choice of $(S,\sim)$ we have that $\can{\Pol(\rel A)}$ acts on $S/{\sim}$ by affine functions over a finite module. We use the symbols $+,\cdot$ for the addition and multiplication in the corresponding ring, and also $+$ for the addition in the module and $\cdot$ for multiplication of elements of the module with elements of the ring. We denote by $1$ the multiplicative identity of the ring, by $-1$ its additive inverse,  and identify their  powers in the additive group with the non-zero integers. The domain of the module is $S/{\sim}$, and we denote the identity element  of its additive group by $[a_0]_\sim$. Pick an arbitrary  element $[a_1]_\sim\neq [a_0]_\sim$ from $S/{\sim}$, and let $m\geq 2$ be its order in the additive group of the module, i.e., the minimal positive number such that $m\cdot [a_1]_\sim=[a_0]_\sim$.
For $i\in\{2,\dots,m-1\}$, let $a_i$ be an arbitrary element such that $[a_i]_{\sim} = i\cdot [a_1]_\sim$.
Let $g\in\can{\Pol(\rel A)}$ be the $m$-ary operation obtained in the preceding paragraph. If the values $g(a_0,\ldots,a_{m-1}),\ldots,g(a_{m-1},a_0,\ldots,a_{m-2})$ lie in pairwise distinct $\Aut(\rel A)$-orbits, then (computing indices modulo $m$) we have that $g([a_0]_\sim,\ldots,[a_{m-1}]_\sim),
\ldots, g([a_{m-1}]_\sim,\ldots,[a_{m+m-1}]_\sim)$ are all equal. 
If on the other hand they do  not, then $g([a_k]_\sim,\ldots,[a_{k+m-1}]_\sim)=  g([a_{k+j}]_\sim,\ldots,[a_{k+j+m-1}]_\sim)$ 
for some $0\leq k<m$ and  $1\leq j<m$. Hence, in either case we may assume the latter equation holds. 
By assumption, $g$ acts on $S/{\sim}$ as an affine map, i.e., as a map of the form  $(x_0,\dots,x_{m-1})\mapsto\sum_{i=0}^{m-1} c_i\cdot x_i$, where $c_0,\dots,c_{m-1}$ are elements of the  ring which sum up to $1$.
We compute (with indices to be read  modulo $m$)
\begin{align*}
    [a_0]_\sim &= g([a_{k+j}]_\sim,\dots,[a_{k+j+m-1}]_\sim)\, +\, (-1)\cdot  g([a_k]_\sim,\dots,[a_{k+m-1}]_\sim)\\
    &= \sum_{i=0}^{m-1} c_i\cdot [a_{k+j+i}]_\sim \, + \, (-1)\cdot   \sum_{i=0}^{m-1} c_i\cdot [a_{k+i}]_\sim \\
    &= \sum_{i=0}^{m-1} c_i\cdot (k+i+j)\cdot [a_1]_\sim\, + \, (-1)\cdot  \sum_{i=0}^{m-1} c_i\cdot (k+i)\cdot [a_1]_\sim\\
    &= \left(\sum_{i=0}^{m-1} c_i\right) \cdot j\cdot  [a_1]_\sim = j\cdot [a_1]_\sim.
\end{align*}
But $j\cdot [a_1]_\sim\neq [a_0]_\sim$ since the order of $[a_1]_\sim$ equals $m>j$, a contradiction.
\end{proof}

\subsection{Proof of Corollary~\ref{cor:explicitcollapse}}\label{sect:corollary}

Before proceeding with the proof of Corollary~\ref{cor:explicitcollapse}, we summarize the currently known characterizations of structures with bounded width.
\begin{corollary}[Theorems 39 and 57 in~\cite{SmoothApproximations}, Theorems 21 and 67 in~\cite{PosetCSP}, Theorem 8.1 in~\cite{BodirskyW12}, Theorem 1 in~\cite{ecsps}, Theorem~\ref{thm:characterization-unary}]\label{recap-bwidth-classifications}
    Let $\rel A$ be a model-complete core structure that has bounded width. Let $\mathscr C$ be the clone of polymorphisms of $\rel A$ that are $\Aut(\rel A)$-canonical. If $\rel A$ is a first-order reduct of:
    \begin{itemize}
        \item the universal homogeneous graph or tournament, or of a unary structure, then $\mathscr C$ contains pseudo-WNU operations modulo $\overline{\Aut(\rel A)}$ of all arities $n\geq 3$;
        \item $(\mathbb N;=)$, the universal homogeneous $K_n$-free graph $\rel H_n$, where $n\geq 3$, or the countably infinite equivalence relation with infinitely many equivalence classes all of infinite size, or the universal homogeneous partial order $\mathbb P$, but not of $(\mathbb Q,<)$, then $\mathscr C$ contains pseudo-totally symmetric operations modulo $\overline{\Aut(\rel A)}$ of all arities.
    \end{itemize}
\end{corollary}
Now, we can finally prove Corollary~\ref{cor:explicitcollapse} from Section~\ref{sect:intro}.

\explicitcollapse*

\begin{proof}
 For the first item, let $\rel A$ be a  first-order reduct  of $\rel{G}$ or $\rel{T}$. Then $\rel A$ has a model-complete core by~\cite{cores}; this model-complete core has the same relational width as $\rel A$. Moreover, by Lemma 17 and Lemma 46 in~\cite{SmoothApproximations}, this model-complete core is again a first-order reduct of $\rel G$ or $\rel T$, respectively, or a one-element structure.
 In the latter case, $\rel A$ has relational width $(1,1)$; hence, we may assume that  $\rel A$ is itself a model-complete core. It then follows from Corollary~\ref{recap-bwidth-classifications} that $\rel A$ has, for every $n\geq 3$, an $\Aut(\rel A)$-canonical polymorphism of arity $n$ that is a pseudo-WNU operation modulo $\overline{\Aut(\rel A)}$.
Since both $\rel{G}$ and $\rel{T}$ are $2$-homogeneous and $3$-bounded, our claim  follows from our Theorem~\ref{thm:canonicalbwidth}

By Lemma 6.7 in~\cite{ReductsUnary}, the set of first-order reducts of unary structures is closed under taking model-complete cores. Let therefore $\rel A$ be a first-order reduct of a unary structure $\rel B$ which is a model-complete core. Then $\rel A$ is also a first-order reduct of a unary structure $\rel C=(C;V_1,\dots,V_r)$, where $V_i$ either contains only one element or is infinite and where $V_i$ and $V_j$ are disjoint for every $1\leq i,j\leq r$ with $i\neq j$.  
Moreover, $\rel C$ is $2$-homogeneous since the orbit of every tuple under $\Aut(\rel B)$ is determined by the unary relations and equalities holding on pairs of its elements. We claim it is $2$-bounded. Indeed, every finite structure in the signature of $\rel C$ embeds into $\rel C$ unless it contains an element contained in two relations $V_i, V_j$ with $1\leq i, j \leq r$, $i\neq j$, or it contains two distinct elements 
contained in a relation $V_i$ for some $1\leq i \leq r$ with $|V_i^{\rel C}|=1$.  Hence, by appeal to Theorem~\ref{thm:canonicalbwidth} and Corollary~\ref{recap-bwidth-classifications}
in the present paper, our claim  holds for this class as well.

If $\rel A$ is a first-order reduct of $\rel{H}_n$, $(\rel{N}; =)$ or $\rel{C}^{\omega}_{\omega}$ that has bounded width and is a model-complete core,
then it has $\Aut(\rel A)$-canonical pseudo-totally symmetric polymorphisms modulo $\overline{\Aut(\rel A)}$ of all arities, by Corollary~\ref{recap-bwidth-classifications} above.
Since $\rel{H}_n$ is $2$-homogeneous and $n$-bounded,  and since both $(\rel{N}; =)$ and $\rel{C}^{\omega}_{\omega}$ are $2$-homogeneous and $3$-bounded, 
the claimed bound  follows.

Finally, a first-order reduct of $\rel{P}$ with bounded width is either homomorphically equivalent to a first-order reduct of $(\rel{Q}; <)$ or it satisfies the algebraic condition in the second item of Theorem~\ref{thm:canonicalbwidth} by Corollary~\ref{recap-bwidth-classifications}.  In the latter case we are done by Theorem~\ref{thm:canonicalbwidth}, in the former we appeal to the syntactical characterization of first-order reducts of $(\rel{Q}; <)$. Indeed, such a structure has bounded width if and only if it is definable by a conjunction of so-called Ord-Horn clauses~\cite{Rydval:2020}. It then follows by~\cite{NebelBueckert} that 
a first-order reduct of $(\rel{Q}; <)$ with bounded width has relational width at most $(2,3)$. The result for $\rel{P}$ follows.
\end{proof}

The following example shows that for some of the structures under consideration, the bounds on relational width provided by  Corollary~\ref{cor:explicitcollapse} are tight.

\begin{example}
To show the tightness of the bound in the case of the  universal homogeneous graph $\rel G=(G;E)$, we exhibit a first-order reduct $\rel A$ such that for all $i\leq j$ with $i\leq 4$, if $1\leq i<4$ or $1\leq j<6$, then there exists a non-trivial, $(i,j)$-minimal instance of $\Csp(\rel A)$ that has no solution.
Let $N:=(G^2\setminus E)\cap{\neq}$.
Consider the first-order reduct $\rel{A}:=(G;R_{=},R_{\neq})$ of $\rel G$, where $R_{=}:=\{(a,b,c,d)\in G^4\mid E(a,b)\wedge E(c,d)\text{ or }N(a,b)\wedge N(c,d)\}$ and $R_{\neq}:=\{(a,b,c,d)\in G^4\mid E(a,b)\wedge N(c,d)\text{ or }N(a,b)\wedge E(c,d)\}$.

It can be seen that $\rel A$ has bounded width, so that Corollary~\ref{cor:explicitcollapse} implies that $\rel A$ has relational width at most $(4,6)$.
The instance $\instance=(\{v_1,\dots,v_4\},\{C_=,C_{\neq}\})$, where $C_=:=\{f\in A^{\{v_1,\dots,v_4\}}\mid f(v_1,\dots,v_4)\in R_=\}$, $C_{\neq}:=\{f\in A^{\{v_1,\dots,v_4\}}\}\mid f(v_1,\dots,v_4)\in R_{\neq}\}$, is non-trivial, and $(i,j)$-minimal for all $i\leq j$ with $1\leq i<4$ since all variables are contained in the scopes of both constraints and the projection of any of the constraints to any subset of $\{v_1,\dots,v_4\}$ containing at most $3$ elements is equal to the projection of $\{f\in A^{\{v_1,\dots,v_4\}}\mid f(v_1)\neq f(v_2), f(v_3)\neq f(v_4)\}$ to this subset. However, $\instance$ 
has no solution.
Moreover, the $(4,5)$-minimal instance equivalent to the instance $\cJ=(\{v_1,\dots,v_6\},\{D_1,D_2,D_3\})$, where $D_1:=\{f\in A^{\{v_1,\dots,v_4\}}\mid f(v_1,\dots,v_4)\in R_{\neq}\}$, $D_2:=\{f\in A^{\{v_3,\dots,v_6\}}\mid f(v_3,\dots,v_6)\in R_{\neq}\}$, and $D_3:=\{f\in A^{\{v_1,v_2,v_5,v_6\}}\mid f(v_1,v_2,v_5,v_6)\in R_=\}$, 
is non-trivial and has no solution.
It follows that the exact relational width of $\rel A$ is $(4,6)$.
\end{example}

The tightness of the bound for first-order reducts of the universal homogeneous tournament can be shown similarly. However, it is an open question whether the bound is tight for first-order reducts of unary structures (for more details, see~\cite{infinitesheep}).

The bound on relational width provided in the second item of Corollary~\ref{cor:explicitcollapse} is tight.
Indeed, let $n\geq 3$, let $\rel  H_n:=(H_n;E)$ be the universal homogeneous $K_n$-free graph, let $N:=((H_n)^2\backslash E)\cap \neq$, and let $\rel A:=(H_n;E,N)$. $\rel A$ is preserved by canonical pseudo-totally symmetric operations modulo $\overline{\Aut(\rel H_n)}$ of all arities and therefore has relational width at most $(2,n)$ by Theorem~\ref{thm:canonicalbwidth} (setting $\rel B:=\rel H_n$). But the non-trivial, $(2,n-1)$-minimal instance $\instance=(\{v_1,\dots,v_n\},\{C_{i,j}\mid 1\leq i\neq j\leq n\})$, where $C_{i,j}:=\{f\in A^{\{v_i,v_j\}}\mid (f(v_i),f(v_j))\in E\}$ for every $1\leq i\neq j\leq n$,
has no solution; moreover, the instance $\cJ=(\{v_1,v_2\},\{\{f\in A^{\{v_1,v_2\}}\mid (f(v_1),f(v_2))\in E\},\{f\in A^{\{v_1,v_2\}}\mid (f(v_1),f(v_2))\in N\}\})$ is non-trivial, $(1,j)$-minimal for every $j\geq 1$ and has no solution either.

For the structures from the third item of Corollary~\ref{cor:explicitcollapse}, the tightness of the bound can be shown similarly.

\subsection{MMSNP}\label{sect:mmsnp}

Monotone Monadic SNP (MMSNP) is a fragment of existential second order logic that was discovered by Feder and Vardi in their seminal paper~\cite{FederVardi}.
MMSNP is defined as the class of formulas of the form

\[ \Phi:=\exists M_1\dots\exists M_r \forall x_1\dots\forall x_k\bigwedge_i \neg(\alpha_i\land\beta_i) \]
where $M_1,\dots,M_r$ are unary predicates, each  $\alpha_i$ is a  conjunction of positive atoms in a signature $\tau$ not containing the equality (called the \emph{input signature}) and each $\beta_i$ is a conjunction of positive or negative existentially quantified unary predicates.
We say that $\Phi$ is \emph{connected} if every conjunction $\alpha_i$ is connected, that is, no $\alpha_i$ can be written as the conjunction of two non-empty formulas which do not share any variables.
A simple syntactic procedure shows that every $\Phi$ is equivalent to a disjunction of connected MMSNP sentences, which can be computed in exponential time.

As in the case of graphs, we call a structure \emph{connected} if it is not the disjoint union of two of its proper substructures. Let $\tau$ be a relational signature,  let  $\sigma$ be a unary signature whose relations are called the \emph{colors}, and let  $\mathcal F$ be a finite set of finite connected  $(\tau\cup\sigma)$-structures, with the property that every element of any structure in $\mathcal F$ has exactly one color. 
We call $\mathcal F$ a \emph{colored obstruction set} in the following.
The problem $\fpp(\mathcal F)$ takes as input a $\tau$-structure $\rel X$ and asks whether there exists a $(\tau\cup\sigma)$-expansion $\rel X^*$ of $\rel X$ whose vertices are all colored with exactly one color and such that for every $\rel F\in\mathcal F$, there exists no homomorphism from $\rel F$ to $\rel X^*$. We say in this case that the coloring $\rel X^*$ is \emph{$\mathcal F$-free} or \emph{obstruction-free}.
It can be seen that connected MMSNP and FPP are equivalent: for every connected MMSNP sentence $\Phi$, there exists a set $\mathcal F$ such that $\rel X\models\Phi$ if and only if $\rel X$ is a yes-instance to $\fpp(\mathcal F)$, for any $\tau$-structure $\rel X$.
Conversely, given any $\mathcal F$, there exists a connected MMSNP sentence $\Phi$ as above.
We call $\mathcal F$ a \emph{colored obstruction set associated with} $\Phi$.

Every connected MMSNP sentence $\Phi$ has an equivalent \emph{normal form} $\Psi$ that can be computed from $\mathcal F$ in double exponential time~\cite[Lemma 4.4]{MMSNP-journal}.
It is shown in~\cite[Definition 4.12]{MMSNP-journal} that for every MMSNP sentence $\Phi$ in normal form, there exists an $\omega$-categorical structure $\rel A_\Phi$ (denoted in~\cite{MMSNP-journal} by $\mathfrak C^\tau_\Phi$) such that for any $\tau$-structure $\rel X$ it is true that $\rel X\models\Phi$ if and only if $\rel X$ admits a homomorphism to $\rel A_\Phi$ (which can equivalently be taken to be injective)~\cite[Lemma 4.13]{MMSNP-journal}.
    
Additionally, $\Phi$ is in \emph{strong} normal form if any identification of two existentially quantified predicates in $\Phi$ yields an inequivalent sentence. 
Finally, we say that $\Phi$ is \emph{precolored} if for every   symbol $M\in\sigma$, there is an associated unary symbol $P_M\in\tau$, and $\Phi$ contains the conjunct $\neg(P_M(x)\land M'(x))$ for every color $M'\in\sigma\setminus\{M\}$.
Note that any precolored sentence in normal form is automatically in strong normal form (identifying two unary predicates $M,M'$ would yield a sentence $\Psi$ containing the conjunct $\neg (P_M(x)\land M(x))$, and $\Phi$ and $\Psi$ can be separated by a $1$-element structure whose sole vertex is in $P_M$).
Any sentence $\Phi$ has a \emph{standard precoloration} obtained by adding for every symbol $M\in\sigma$ a predicate $P_M$ to $\tau$ and for every $M'\in\sigma\setminus\{M\}$ the conjunct $\neg(P_M(x)\land M'(x))$ to $\Phi$. The colored obstruction set of the standard precoloration of $\Phi$ consists of the obstruction set for $\Phi$ together with one-element obstructions $\rel F$ whose sole vertex belongs to $P_M$ and $M'$ in $\rel F$ and all other relational symbols are interpreted as empty relations.

A useful (but imprecise) parallel between MMSNP and finite-domain CSPs is the following. Connected MMSNP sentences in normal form, in strong normal form, and in precolored normal form have the same relationship as do finite structures, finite cores, and expansions of cores by all singleton unary relations.

When $\Phi$ is in strong normal formal or precolored, $\rel A_\Phi$ can additionally be chosen to have the following properties:
\begin{enumerate}[label=(\roman*)]
    \item If $\Phi$ is precolored, then the orbits of the elements of  $\rel A_{\Phi}$ under $\Aut(\rel A_{\Phi})$ correspond to the colors of ${\Phi}$ and  coincide with the interpretations of the corresponding symbols in $\tau$~\cite[Lemma 6.2]{MMSNP-journal}. 
    In particular, the action of  $\cann[1]{\Pol(\rel A_{\Phi})}$ on $\Aut(\rel A_{\Phi})$-orbits of elements is idempotent~\cite[Proposition 7.2]{MMSNP-journal}.\label{itm:colors-orbits}
    \item The expansion of $\rel A_\Phi$ by a generic linear order $<$ is a Ramsey structure~\cite[Corollary 5.9]{MMSNP-journal}.
    By Theorem~\ref{thm:canonical-homo}, every $f\in\Pol(\rel A_{\Phi})$ locally
    interpolates an operation $g\in\cann[1]{\Pol(\rel A_{\Phi})}$, and every $f$ diagonally interpolates an operation $f'$ that is diagonally canonical with respect to $\Aut(\rel A_{\Phi},<)$.\label{itm:interpolation-mmsnp}
\end{enumerate}

We finally solve the Datalog-rewritability problem for MMSNP and prove that a precolored connected sentence $\Phi$ in normal form is equivalent to a Datalog program if and only if the action of $\cann[1]{\Pol(\rel A_{\Phi})}$ on ${\Aut(\rel A_{\Phi})}$-orbits of elements is not equationally affine.

We will need the following results from Ramsey theory.

\begin{theorem}[Consequence of Theorem 2.11 from~\cite{HubickaNesetril}]\label{thm:HN}
    Let $\Phi$ be an MMSNP sentence in normal form. Let $\rel A^*$ be the expansion of $\rel A_{\Phi}$ by all pp-definable relations.
    Then there exists a countably infinite homogeneous structure $(\rel H,<)$ whose finite substructures are exactly the finite substructures of structures satisfying $\Phi$ and expanded by all pp-definable relations and by an arbitrary linear order.
    Moreover, $(\rel H,<)$ is Ramsey.
\end{theorem}

\begin{proof}
    See Appendix A in~\cite{MMSNP-journal} for the translation of Theorem 2.11 from~\cite{HubickaNesetril} into the MMSNP framework.
\end{proof}

The following proposition is proved in~\cite[Lemma 7.5]{MMSNP-journal} in the case where $m=2$.
We give the proof for the convenience of the reader.
\begin{proposition}\label{prop:mmsnp-property}
    Let $\Phi$ be a precolored MMSNP sentence in normal form and let $m\geq 1$.
    There exist self-embeddings $e_1,\dots,e_m$ of $\rel A_{\Phi}$ such that
    $(e_{i_1}(a_1),\dots,e_{i_m}(a_m))$
    and
    $(e_{j_1}(b_1),\dots,e_{j_m}(b_m))$
    are in the same orbit under $\Aut(\rel A_{\Phi},<)$ provided that:
    \begin{itemize}
        \item $a_k$ and $b_k$ are in the same color for all $k\in\{1,\dots,m\}$
        \item $a_k$ and $a_\ell$ are in distinct colors for all $k\neq \ell$,
        \item  $\{i_1,\dots,i_m\}=\{j_1,\dots,j_m\}=\{1,\dots,m\}$.
    \end{itemize}
\end{proposition}
\begin{proof}
    Let $(\rel H,<)$ be the Ramsey structure from Theorem~\ref{thm:HN}. 
    There exists a homomorphism $h\colon \rel H\to\rel A^*$ since the reduct of $\rel H$ to the input signature of $\Phi$ satisfies $\Phi$ by definition and hence, it admits a homomorphism to $\rel A_{\Phi}$; this homomorphism also preserves the additional relations in the signature of $\rel H$ and $\rel A^*$ since these relations are obtained by pp-definitions from relations in the input signature of $\Phi$. Theorem~\ref{thm:canonical-homo} yields that $h$ can moreover be assumed to be canonical from $\Aut(\rel H,<)$ to $\Aut(\rel A_{\Phi},<)$.
    
    Let $\rel B$ be 
    the expansion of $\{1,\dots,m\}\times\rel A_\Phi$, the disjoint union of $m$ copies of $\rel A_\Phi$, by all pp-definable relations.
    Endow $\rel B$ with a linear order that is convex with respect to the colors.
    Since $\Phi$ is connected, and since $\rel A_\Phi$ satisfies $\Phi$, we have that $\rel B$ satisfies $\Phi$.
    All the finite substructures of $(\rel B,<)$ then embed into $(\rel H,<)$, and by compactness there exists an embedding $e'$ of $(\rel B,<)$ into $(\rel H,<)$.
    Let $e_i(x):=(h\circ e')(i,x)$.
    
    To check that these self-embeddings satisfy the required properties, let $a_1,b_1,\dots,a_m,b_m$ and $i_1,j_1,\dots,i_m,j_m$ be as in the statement.
    Note that since $a_k$ and $b_k$ are in the same color for all $k$, they are in the same orbit in $\rel A_{\Phi}$ by the fact that $\Phi$ is precolored~\ref{itm:colors-orbits}.
    They in particular satisfy the same primitive positive formulas, meaning that $e'$ maps $(i_k,a_k)$ and $(j_k,b_k)$ to elements of $\rel H$ that satisfy the same atomic formulas, and that are therefore in the same orbit as $\rel H$ is homogeneous.
    Note moreover that by definition of the order on $\rel B$, $(i_k,a_k)<(i_\ell,a_\ell)$ if and only if $(j_k,b_k)<(j_\ell,b_\ell)$, and that no other atomic relation holds within the tuples $((i_1,a_1),\dots,(i_m,a_m))$ and $((j_1,b_1),\dots,(j_m,b_m))$.
    Thus, the required tuples are in the same orbit in $(\rel H,<)$, by the  homogeneity of $(\rel H,<)$.
    Since $h$ is canonical, we obtain that their image under $h$ is in the same orbit in $(\rel A_{\Phi},<)$.
\end{proof}

\begin{restatable}{lemma}{canonisationmmsnp}\label{lem:canonisation-mmsnp}
Let $(S,\sim)$ be a subfactor of $\cann[1]{\Pol(\rel A_{\Phi})}$ with $\Aut(\rel A_{\Phi})$-invariant $\sim$-classes. Let $m\geq 2$, and let $f\in\Pol(\rel A_{\Phi})$ be such that for all $a_1,\ldots,a_m\in A_{\Phi}$, if the entries
of the tuple $(f(a_1,\dots,a_m),f(a_2,\dots,a_m,a_1),\dots,f(a_m,a_1,\dots,a_{m-1}))$ all belong to different colors, then the tuple intersects at most one $\sim$-class.
    Let $O_0,\dots,O_{m-1}\in S$ be pairwise distinct orbits under $\Aut(\rel A_{\Phi})$.
    There exists $g\in\cann[1]{\Pol(\rel A_{\Phi})}$ that is locally interpolated by $f$ and that satisfies
    \[g(O_k,\dots,O_{k+m-1}) \sim g(O_{j+k},\dots,O_{j+k+m-1})\tag{$\star$}\]
    for some $0\leq k<m$ and $1\leq j<m$ (where the indices are computed modulo $m$).
\end{restatable}
\begin{proof}
    Recall that by \ref{itm:interpolation-mmsnp} the expansion of  $\rel A_{\Phi}$ by a generic linear order is a Ramsey structure.
    By Theorem~\ref{thm:canonical-homo}, $f$ diagonally interpolates a function $g\in\Pol(\rel A_{\Phi})$ with the same properties and which  is diagonally canonical with respect to $\Aut(\rel A_{\Phi},<)$, and without loss of generality we can therefore assume that $f$ is itself diagonally canonical.
    
    Let $e_0,\dots,e_{m-1}$ be self-embeddings of $\rel A_{\Phi}$ with the properties stated in Proposition~\ref{prop:mmsnp-property}.
    Consider $f'(x_0,\dots,x_{m-1}):=f(e_0x_0,\dots,e_{m-1}x_{m-1})$,
    and note that $f'$ is $1$-canonical when restricted to $m$-tuples where all entries are in pairwise distinct orbits.
    Let $g \in \Pol(\rel A_{\Phi})$ be a function that is diagonally interpolated by $f'$ and which is diagonally canonical with respect to $\Aut(\rel A_{\Phi},<)$.
    In particular $g\in\cann[1]{\Pol(\rel A_{\Phi})}$ and
    $g(O_k,\dots,O_{k+m-1})$ and $f'(O_k,\dots,O_{k+m-1})$ are in $S$ and $\sim$-equivalent for all $k$.
        
    As in the proof of Theorem~\ref{thm:characterization-unary},
    there are suitable $0\leq k<m$ and  $1\leq j<m$ such that
    \[ f(e_kO_k,\dots,e_{k+m-1}O_{k+m-1})\sim f(e_{k+j}O_{k+j},\dots,e_{k+j+m-1}O_{j+k+m-1}) \]
    holds, where indices are computed modulo $m$.
    Then 
    \begin{align*}
        g(O_k,\dots,O_{k+m-1}) &\sim f(e_0O_k,\dots,e_{m-1}O_{k+m-1})\\
        &\sim f(e_{k}O_k,\dots,e_{k+m-1}O_{k+m-1})&\tag{$\ddag$}\\
        &\sim f(e_{k+j}O_{k+j},\dots,e_{k+j+m-1}O_{k+j+m-1})\\
        &\sim f(e_0O_{k+j},\dots,e_{m-1}O_{k+j+m-1})&\tag{$\ddag$}\\
        &\sim g(O_{k+j},\dots,O_{k+j+m-1}),
    \end{align*}
    where the equivalences marked $(\ddag)$ hold by the fact that $f$ is diagonally canonical with respect to $\Aut(\rel A_{\Phi},<)$ and by Proposition~\ref{prop:mmsnp-property}.
\end{proof}

Let $\Phi$ be an MMSNP sentence in normal form with input signature $\tau$ and let $\mathcal F$ be the associated colored obstruction set.
Note that $\rel A_\Phi$ satisfies $\Phi$, so that there exists a coloring of $\rel A_\Phi$ that is $\mathcal F$-free. Fix such a coloring $\rel A_\Phi^*$.
For each $\tau$-reduct $\rel F^\tau$ of a structure $\rel F\in\mathcal F$ of size $k$, fix an enumeration of the domain and assume without loss of generality that the domain is $\{1,\dots,k\}$. Let $\rel B_\Phi$ be the finite structure whose domain is the set of colors of $\Phi$, with one $k$-ary relation $R_{\rel F^\tau}$ for each $\tau$-reduct of an obstruction $\rel F\in\mathcal F$ of size $k$, and where the interpretation of $R_{\rel F^\tau}$ in $\rel B_\Phi$ contains all tuples $(C_1,\dots,C_k)$ of colors such that there exists a homomorphism $h\colon \rel F^\tau\to\rel A_\Phi$ and such that $h(i)$ is in the color $C_i$ in $\rel A^*_{\Phi}$ for all $i$.

Let $\rel A$ be a relational structure, and let $\instance=(\V,\constraints)$ be an instance of $\Csp(\rel A)$. The \emph{canonical database} of $\instance$ is the relational structure $\rel D$ in the signature of $\rel A$ such that for every relational symbol $R$ of arity $k\geq 1$, a tuple $(v_1,\dots,v_k)\in\V^k$ is contained in $R^{\rel D}$ if $\constraints$ contains the constraint $\{f\in A^{\{v_1,\dots,v_k\}}\mid f(v_1,\dots,v_k)\in R^{\rel A}\}$.

The following theorem gives a characterization of Datalog-rewritability for precolored normal forms. The proof is similar to that of Theorem~\ref{thm:characterization-unary}.
\begin{restatable}{theorem}{characterizationmmsnp}\label{thm:characterization-mmsnp}
    Let $\Phi$ be a precolored connected MMSNP $\tau$-sentence in normal form, and let $\ell$ be the maximum size of a structure in a minimal colored obstruction set associated with $\Phi$.
    The following are equivalent:
    \begin{enumerate}
        \item $\neg\Phi$ is equivalent to a Datalog program;
        \item $\Pol(\rel A_{\Phi})$ does not have a uniformly continuous minion homomorphism to an affine clone;
        \item The action of  $\cann[1]{\Pol(\rel A_{\Phi})}$ on ${\Aut(\rel A_{\Phi})}$-orbits of elements is not  equationally affine;
        \item $\Pol(\rel B_\Phi)$ is not equationally affine;
        \item $\rel A_{\Phi}$ has relational width at most $(2,\max(3,\ell))$.
    \end{enumerate}
\end{restatable}

\begin{proof}
    We prove the implications (1)$\Rightarrow$(2)$\Rightarrow$(3)$\Rightarrow$(4)$\Rightarrow$(5)$\Rightarrow$(1).
    
    (1) implies (2) by the implication from (\ref{itm:datalog}') to (\ref{itm:minion-affine}) in Theorem~\ref{thm:characterization-bwidth}.

    (2) implies (3). We do the proof by contraposition.
    The proof is essentially the same as in the case of reducts of unary structures (Theorem~\ref{thm:characterization-unary}).
    Suppose that $\cann[1]{\Pol(\rel A_{\Phi})}\actson \rel A_{\Phi}/{\Aut(\rel A_{\Phi})}$ is equationally affine
    and let $(S,\sim)$ be as given by Proposition~\ref{prop:eq-affine_subfactor}.
    
    Let $m\geq 2$ and let $R$ be an $m$-ary cyclic relation invariant under $\Pol(\rel A_{\Phi})$ and containing a tuple $(a_1,\dots,a_{m})$ whose entries are pairwise distinct.
    By Proposition~\ref{prop:approximation-dichotomy}, either the linkedness congruence of $R$ defines an approximation of $\sim$, or $R$ contains a pseudoloop modulo $\Aut(\rel A_{\Phi})$.
    In the first case, the approximation is smooth by the same argument as in the proof of Theorem~\ref{thm:characterization-unary} and we obtain a uniformly continuous minion homomorphism from $\Pol(\rel A_{\Phi})$ to an affine clone by Theorem~\ref{thm:fundamental}.
    
    So let us assume that for all $m\geq 2$, every such relation $R$ contains a pseudoloop.
    By applying Lemma~\ref{lem:n-ary-weird-cyclic} with $P$ being the set of $m$-tuples whose entries belong to pairwise distinct colors,
    we obtain a polymorphism $f$ such that
    for all $a_1,\dots,a_{m}$,
    if $f(a_1,\dots,a_{m}),\dots,f(a_{m},a_0,\dots,a_{m-1})$ are in pairwise distinct colors, then they intersect at most one $\sim$-class.
    As in the proof of Theorem~\ref{thm:characterization-unary}, pick an arbitrary $a_1\in S$ such that $[a_1]_{\sim}$ is not the zero element of the module $S/{\sim}$.
    Let $m\geq 2$ be its order, and let $O_i$ be the orbit of $i\cdot [a_1]_{\sim}$, for $i\in\{0,1,\dots,m-1\}$.
    By Lemma~\ref{lem:canonisation-mmsnp},
    we obtain $g\in\cann[1]{\Pol(\rel A_{\Phi})}$ such that
    \[ g(O_k,\dots,O_{k+m-1})\sim g(O_{j+k},\dots,O_{j+k+m-1})\]
    for some $k\in\{0,\dots,m-1\}$ and $j\in\{1,\dots,m-1\}$.
    The same computation as in Theorem~\ref{thm:characterization-unary} then gives a contradiction and concludes the proof.
    
    (3) implies (4). Since $\Phi$ is precolored, orbits under $\Aut(\rel A_\Phi)$ corresponds to the colors of $\Phi$, i.e., to the domain of $\rel B_\Phi$ (by~\ref{itm:colors-orbits}).
    Then under the bijection between colors and orbits, $\can[1]{\Pol(\rel A_\Phi)}\actson \rel A_\Phi/{\Aut(\rel A_\Phi)}$ is a subset of $\Pol(\rel B_\Phi)$.
    Thus if the smaller clone is not equationally affine, the larger clone is also not.
    
    (4) implies (5). Let $\mathcal F$ be a minimal colored obstruction set associated with $\Phi$.
    Suppose that $\Pol(\rel B_\Phi)$ is not equationally affine. Then $\Pol(\rel B_{\Phi})$ does not have a minion homomorphism to an affine clone and moreover, all colors (i.e., elements of the domain of $\rel B_\Phi$) are pp-definable in $\rel B_\Phi$ since $\Phi$ is precolored. Hence, by the implication from 
    (\ref{itm:core-affine}) to (\ref{itm:width-2-3}) in Theorem~\ref{thm:characterization-bwidth}, $\rel B_\Phi$ has relational width at most $(2,3)$.
    It is proven in~\cite[Theorem 7.3]{MMSNP-journal} that the following is a reduction from $\Csp(\rel A_\Phi)$ to $\Csp(\rel B_\Phi)$: given an instance $\instance=(\V,\constraints)$ of $\Csp(\rel A_\Phi)$, define an instance $\mathcal J$ of $\Csp(\rel B_\Phi)$ with the same variables and which, for every obstruction $\rel F\in\mathcal F$ of size $k\geq 1$ and for every homomorphism $h$ from $\rel F^\tau$ to the canonical database of $\instance$, contains the constraint $\{f\in B_\Phi^{\{h(1),\dots,h(k)\}}\mid (f(h(1)),\dots,f(h(k)))\in R_{\rel F^\tau}\}$. 
    If $\instance$ is equivalent to a non-trivial $(2,\max(3,\ell))$-minimal instance, then $\mathcal J$ is equivalent to a $(2,3)$-minimal instance.
    It follows that $\rel A_\Phi$ has relational width at most $(2,\max(3,\ell))$.
        
    (5) implies (1). By the implication from (\ref{itm:bw}) to (\ref{itm:datalog}') in Theorem~\ref{thm:characterization-bwidth}, (5) implies that the class of finite structures that do not have a homomorphism to $\rel A_\Phi$ is definable by a Datalog program.
    This class is by construction the class of finite models of $\neg\Phi$, which proves (1).
\end{proof}

The bound from item (5) of Theorem~\ref{thm:characterization-mmsnp} can be reached for any $\ell\geq 3$ which in particular implies that the relational width of structures $\rel A_{\Phi}$ associated to precolored MMSNP sentences in normal form is not uniformly bounded. Indeed, let $\ell\geq 3$, let $\tau$ consist of a single binary relational symbol $E$, and let $\Phi$ be the MMSNP sentence whose colored obstruction set consists of a complete graph on $\ell$ vertices, all of which belong to a single color. 
Then the structure $\rel A_{\Phi}$ is the universal homogeneous $K_n$-free graph, and we showed in Section~\ref{sect:corollary} that the exact relational width of $\rel A_{\Phi}$ is $(2,\ell)$.

Note that the characterization in Theorem~\ref{thm:characterization-mmsnp} only holds for precolored sentences in normal form.
We show below how to characterize Datalog also for arbitrary MMSNP sentences in normal form.
The steps involved in the proof to go from normal form to precolored normal form reflect a similar situation as for finite domain CSPs, when going from arbitrary finite structures to their cores: given a finite structure $\rel A$ and its core $\rel B$, we have that $\Csp(\rel A)$ is in Datalog if and only if $\Pol(\rel B)$ is not equationally affine if and only if $\Pol(\rel A)$ does not admit a minion homomorphism to an affine clone (see Theorem~\ref{thm:characterization-bwidth}).

On the algebraic side, we first show in Proposition~\ref{prop:reduction-precolored} that the existence of a uniformly continuous minion homomorphism from $\Pol(\rel A_\Phi)$ or $\Pol(\rel A_\Theta)$ to an affine clone does not depend on whether we consider an MMSNP formula in normal form $\Phi$ or the standard precoloration $\Theta$ of an equivalent formula in strong normal form.
The same step was shown in~\cite{MMSNP, MMSNP-journal} for the P/NP-complete dichotomy, with $\Proj$ replacing affine clones.
We then show in Lemma~\ref{lemma:MMSNP-core-singleton} that $\rel B_\Theta$ is the expansion of the core of $\rel B_\Phi$ by singleton unary relations, which allows us to use the characterization of bounded width in Theorem~\ref{thm:characterization-bwidth}.
Finally we show in Lemma~\ref{lemma:BPhi_Datalog} that $\neg\Theta$ is equivalent to a Datalog program if and only if $\neg\Phi$ is.

\begin{restatable}{proposition}{reductionprecolored}
\label{prop:reduction-precolored}
    Let $\Phi$ be an MMSNP sentence in normal form, let $\Psi$ be an equivalent formula in strong normal form, and let $\Theta$ be the standard precoloration of $\Psi$.
    There is a uniformly continuous minion homomorphism from $\Pol(\rel A_{\Phi})$ to an affine clone if and only if there is a uniformly continuous minion homomorphism from $\Pol(\rel A_{\Theta})$ to an affine clone.
\end{restatable}

\begin{proof}
    Since $\Phi$ and $\Psi$ are equivalent, $\rel A_\Phi$ and $\rel A_\Psi$ are homomorphically equivalent (this can be seen using a compactness argument, cf.~\cite[Theorem 5.15]{MMSNP-journal}).
    Therefore, there exist uniformly continuous minion homomorphisms $\Pol(\rel A_\Phi)\to\Pol(\rel A_\Psi)$ and $\Pol(\rel A_\Psi)\to\Pol(\rel A_\Phi)$.
    
    It is shown in~\cite[Theorem 6.9]{MMSNP-journal} that $\Pol(\rel A_{\Theta})$ has a uniformly continuous minion homomorphism to $\Pol(\rel A_{\Psi})$ and that $\Pol(\rel A_{\Psi},\neq)$ has a uniformly continuous minion homomorphism to $\Pol(\rel A_{\Theta})$.
    Thus, it suffices to show that if $\Pol(\rel A_{\Psi},\neq)$ has a uniformly continuous minion homomorphism to an affine clone, then so does $\Pol(\rel A_{\Psi})$.
    
    Let $p\geq 2$ be prime and let $R_0$ and $R_1$ be the relations from item (\ref{itm:pp-construction-zp}') in Theorem~\ref{thm:characterization-bwidth}.
    By the equivalence of items (\ref{itm:minion-affine}) and (\ref{itm:pp-construction-zp}') in Theorem~\ref{thm:characterization-bwidth}, it is enough to show that if $(\rel A_{\Psi},\neq)$ pp-constructs $(\rel Z_p; R_0,R_1)$, then so does $\rel A_{\Psi}$.
    
    Suppose that $(\rel Z_p;R_0,R_1)$ has a pp-construction in $(\rel A_{\Psi},\neq)$.
    Thus, there is $n\geq 1$ and pp-formulas $\phi_0(x,y,z),\phi_1(x,y,z)$ defining relations $S_0,S_1$ such that $(A^n;S_0,S_1)$ and $(\rel Z_p;R_0,R_1)$ are homomorphically equivalent; we take $n$ to be minimal with the property that such pp-formulas exist. Let $x=(x_1,\ldots,x_n), y=(y_1,\ldots,y_n), z=(z_1,\ldots,z_n).$
    Since $R_0$ and $R_1$ are totally symmetric relations (i.e., the order of the entries in a tuple does not affect its membership into any of $R_0$ or $R_1$), we can assume that $S_0$ and $S_1$ are, too, and that the formulas pp-defining them are syntactically invariant under permutation of the blocks of variables $x$, $y$, and $z$.
    
    We first claim that $\phi_i$ does not contain any inequality atom $x_r\neq y_r$ for $r\in\{1,\dots,n\}$ (so that by symmetry, also $y_r\neq z_r$ and $x_r\neq z_r$ do not appear).
    Let $h\colon(\rel Z_p;R_0,R_1)\to (A^n;S_0,S_1)$ be a homomorphism.
    Since $(0,0,0)\in R_0$, we have that $(h(0),h(0),h(0))$ satisfies $\phi_0$, and therefore the listed inequality atoms cannot appear.
    The same holds for $\phi_1$, by considering $(h(0),h(0),h(1))$ and its permutations.
    
    Moreover, the only possible equality atoms in $\phi_0$  are of the form $u_r=v_r$ for some $r\in\{1,\dots,n\}$ and $u,v\in\{x,y,z\}$.
    Indeed, if $\phi_0$ contains $u_r=v_s$ for $r\neq s$, then by symmetry of $\phi_0$ we obtain by transitivity that $\phi_0$ entails $x_r=x_s$.
    One could obtain a pp-construction with smaller dimension by adding all the equalities $x_r=x_s,y_r=y_s,z_r=z_s$ to $\phi_0$ and $\phi_1$ and existentially quantifying the $s$th coordinate of the two formulas.
    The same reasoning applies to $\phi_1$.

    Let $\psi_i$ be the formula obtained from $\phi_i$ by removing the possible inequality literals, and let $T_i$ be defined by $\psi_i$ in $\rel A_{\Phi}$.
    We claim that $(A^n;T_0,T_1)$ and $(A^n;S_0,S_1)$ are homomorphically equivalent, which concludes the proof.
    Since $\phi_i$ implies $\psi_i$, we have that $(A^n;S_0,S_1)$ is a (non-induced) substructure of $(A^n;T_0,T_1)$, and therefore it homomorphically maps to $(A^n;T_0,T_1)$ by the identity map.
    For the other direction, we prove the result by compactness and show that every finite substructure $\rel B$ of $(A^n;T_0,T_1)$ has a homomorphism to $(A^n;S_0,S_1)$.
    Let $b^1,\dots,b^m$ be the elements of $\rel B$, where $b^i=(b^i_1,\ldots,b^i_n)$.
    Let $\rel C$ be the $\tau$-structure over at most $n\cdot m$ elements $\{c^i_r \mid i\in\{1,\dots,m\},r\in\{1,\dots,n\}\}$ such that:
    \begin{itemize}
        \item $c^i_r$ and $c^j_s$ are taken to be equal if, and only if, $b^i_r$ and $b^j_s$ are connected by a sequence of equalities coming from $\phi_0$ and $\phi_1$. By the claims above, this is only possible if $r=s$.
        \item the relations of $\rel C$ are defined by pulling back the relations from $\rel A_{\Psi}$ under the map $\pi\colon c^i_j\mapsto b^i_j$.
    \end{itemize}
    Note that $\pi$ is a homomorphism $\rel C\to\rel A_{\Psi}$, and therefore $\rel C$ admits an injective homomorphism $g$ to $\rel A_{\Psi}$. Let $c^i=(c^i_1,\ldots,c^i_n)$ for $i\in\{1,\ldots,n\}$.
    We claim that if $(b^i,b^j,b^k)\in T_0$ then $(g(c^i),g(c^j),g(c^k))\in S_0$.
    Indeed, suppose that $(b^i,b^j,b^k)$ satisfies $\psi_0$.
    Then by construction $(c^i,c^j,c^k)$ satisfies $\psi_0$ in $\rel C$, and thus $(g(c^i),g(c^j),g(c^k))$ satisfies $\psi_0$ in $\rel A_{\Psi}$.
    Moreover, by injectivity of $g$, we have $g(c^i_r)\neq g(c^j_s)$ as long as $r\neq s$.
    Consider any inequality atom in $\phi_0$.
    By our first claim, it can only be of the form $x_r\neq y_s$ for some $r\neq s$,
    and therefore it is satisfied by $(g(c^i),g(c^j),g(c^k))$.
    Thus, $(g(c^i),g(c^j),g(c^k))$ satisfies $\phi_0$.
    The same reasoning for $\phi_1$ shows that $g$ induces a homomorphism $\rel B\to (A^n;S_0,S_1)$ by mapping $b^i$ to $g(c^i)$.
\end{proof}

\begin{lemma}\label{lemma:MMSNP-core-singleton}
    Let $\Phi$ be an MMSNP sentence in normal form, let $\Psi$ be an MMSNP sentence in strong normal form equivalent to $\Phi$, and let $\Theta$ be the standard precoloration of $\Psi$. Then $\rel B_\Theta$ is the expansion of the core of $\rel B_\Phi$ by all unary singleton relations.
\end{lemma}

\begin{proof}
    Note that since $\Phi$ and $\Psi$ are equivalent, the structures $\rel A_\Phi$ and $\rel A_\Psi$ are homomorphically equivalent.
    The homomorphisms realizing the equivalence can be taken to be 1-canonical without loss of generality (by~\ref{itm:interpolation-mmsnp}), and thus they induce a homomorphic equivalence between $\rel B_\Phi$ and $\rel B_\Psi$.
    Moreover, $\rel B_\Psi$ is a core: if $h$ is an endomorphism of $\rel B_\Psi$ such that $h(M)=h(M')$ for some colors $M\neq M'$, one can identify in $\Psi$ the colors according to $h$ and obtain an equivalent sentence, contradicting the fact that $\Psi$ is in strong normal form.
    Recall that the standard precoloration of $\Psi$ is obtained by adding 1-element obstructions for each pair of colors $M\neq M'$.
    Since in $\rel A_{\Theta}$ the relation $P_M$ equals the color $M$ (by~\ref{itm:colors-orbits}), these obstructions yield in $\rel B_\Theta$ unary relations containing a single element.
    Thus, $\rel B_\Theta$ is the expansion of the core of $\rel B_\Phi$ by all unary singleton relations.
\end{proof}

\begin{lemma}\label{lemma:BPhi_Datalog}
    Let $\Phi$ be an MMSNP sentence in strong normal form and let $\Theta$ be its standard precoloration.
    If $\neg\Theta$ is equivalent to a Datalog program, so is $\neg\Phi$.
\end{lemma}

\begin{proof}
    Note that if $\rel X$ is a structure in the signature of $\Phi$, it satisfies $\Phi$ if and only if its expansion by the relations $P_M$ for every color $M$ interpreted as empty relations satisfies $\Theta$.
    Thus, we obtain that $\neg\Phi$ is equivalent to a Datalog program obtained by taking a program for $\neg\Theta$ and removing all rules involving the extra predicates $P_M$.
\end{proof}

\begin{corollary}\label{cor:finite-structure-MMSNP}
    Let $\Phi$ be an MMSNP sentence in normal form.
    Then $\neg\Phi$ is equivalent to a Datalog program if, and only if, $\Pol(\rel B_\Phi)$ does not admit a minion homomorphism to an affine clone.
\end{corollary}
\begin{proof}
    Let $\Psi$ be an MMSNP sentence in strong normal form equivalent to $\Phi$, and let $\Theta$ be the standard precoloration of $\Psi$.
    By Lemma~\ref{lemma:MMSNP-core-singleton}, we have that $\rel B_\Theta$ is the expansion of the core of $\rel B_\Phi$ by all singleton unary relations.
    
    If $\Pol(\rel B_\Phi)$ does not admit a minion homomorphism to an affine clone, then $\Pol(\rel B_\Theta)$ is not equationally affine by the implication from~(\ref{itm:minion-affine}) to~(\ref{itm:core-affine}) in Theorem~\ref{thm:characterization-bwidth}, and by Theorem~\ref{thm:characterization-mmsnp}, $\neg\Theta$ is equivalent to a Datalog program.
    By Lemma~\ref{lemma:BPhi_Datalog}, we obtain that $\neg\Phi$ is equivalent to a Datalog program.

    Conversely, suppose that $\Pol(\rel B_\Phi)$ has a minion homomorphism to an affine clone.
    Then $\Pol(\rel B_\Theta)$ is equationally affine by the implication from (\ref{itm:core-affine}) to (\ref{itm:minion-affine}) in Theorem~\ref{thm:characterization-bwidth}, therefore by Theorem~\ref{thm:characterization-mmsnp} there is a uniformly continuous minion homomorphism from $\Pol(\rel A_\Theta)$ to an affine clone.
    By Proposition~\ref{prop:reduction-precolored}, there is a uniformly continuous minion homomorphism from $\Pol(\rel A_\Psi)$ to an affine clone.
    Since $\rel A_\Psi$ and $\rel A_\Phi$ are homomorphically equivalent, we obtain a uniformly continuous minion homomorphism from $\Pol(\rel A_\Phi)$ to an affine clone.
    Finally, by the implication from (\ref{itm:bw}) to (\ref{itm:minion-affine}) in Theorem~\ref{thm:characterization-bwidth}, this implies that $\rel A_\Phi$ does not have bounded width, i.e., $\neg\Phi$ is not equivalent to a Datalog program by the same reasoning as in the proof of (5) implies (1) in Theorem~\ref{thm:characterization-mmsnp}.
\end{proof}

This finally allows us to obtain Theorem~\ref{thm:datalogrewritability} from the introduction.

\boundedwidthmmsnp*
\begin{proof}
    Let $\Phi$ be an MMSNP sentence, which is equivalent to a disjunction $
    \Phi_1\lor\cdots\lor\Phi_p$ of connected MMSNP sentences, and this decomposition can be computed in exponential time~\cite[Proposition 3.2]{MMSNP}.
    Each $\Phi_i$ has size polynomial in $\Phi$.
    Moreover, if $p$ is minimal then $\neg\Phi$ is equivalent to a Datalog program if and only if every $\neg\Phi_i$ is equivalent to a Datalog program (see, e.g., Proposition~3.3 in~\cite{MMSNP-journal}, for a proof of a similar fact).
    Such a minimal set $\{\Phi_1,\dots,\Phi_p\}$ of sentences can be computed in exponential time, given $\Phi$ as input.
    After having computed any set $\{\Phi_1,\dots,\Phi_p\}$ whose disjunction is equivalent to $\Phi$, it suffices to iterate the following procedure: for any $i,j\in\{1,\dots,p\}$, check whether $\Phi_i$ implies $\Phi_j$ (by Theorem 5.15 in~\cite{MMSNP-journal}, this problem is in NP).
    If $\Phi_i$ implies $\Phi_j$ for some $i\neq j$, then remove $\Phi_i$ and continue.
    Otherwise we claim $\{\Phi_1,\dots,\Phi_p\}$ is minimal.
    Indeed, no $\Phi_i$ distinct from $\Phi_{j_1},\ldots,\Phi_{j_k}$ implies any disjunction $\Phi_{j_1}\lor\cdots\lor\Phi_{j_k}$:
    by taking for each $j\in\{j_1,\dots,j_k\}$ a structure $\rel X_j$ witnessing that $\Phi_i$ does not imply $\Phi_j$ (i.e., $\rel X_j$ satisfies $\Phi_i$ but not $\Phi_j$), then the disjoint union $\rel X_{j_1}\cup\cdots\cup \rel X_{j_k}$ witnesses that $\Phi_i$ does not imply the disjunction.
    
    Fix $i\in\{1,\dots,p\}$.
    Let $\Psi_i$ be a normal form associated with $\Phi_i$, which can be computed in double exponential time, and where $\Psi$ itself can be of size doubly exponential in the size of $\Phi_i$.
    By Corollary~\ref{cor:finite-structure-MMSNP}, $\neg\Psi_i$ is equivalent to a Datalog program if, and only if, $\Pol(\rel B_{\Psi_i})$ does not admit a minion homomorphism to an affine clone.
    Deciding this property is in NP~\cite[Corollary 6.8]{MetaChenLarose}.
    We obtain overall a 2NExpTime algorithm.
    The complexity lower bound is Theorem~18 in~\cite{BourhisLutz}.
\end{proof}

\bibliographystyle{plainurl}
\bibliography{references}

\end{document}

%% file: commands.tex
\newcommand{\ignore}[1]{}
\newcommand{\aA}{\mathbf A}
\newcommand{\aB}{\mathbf B}

\newcommand{\aR}{\mathbf R}

\newtheorem{theorem}{Theorem}[section]
\newtheorem{corollary}[theorem]{Corollary}
\newtheorem{lemma}[theorem]{Lemma}
\newtheorem{proposition}[theorem]{Proposition}
\theoremstyle{definition}
\newtheorem{example}[theorem]{Example}
\newtheorem{definition}[theorem]{Definition}

\newcommand{\cl}[1]{\langle #1\rangle_\cD}

\DeclareMathOperator{\canonical}{can}
\newcommand{\cann}[2][n]{#2_{#1}^{\canonical}}
\newcommand{\can}[2][]{#2_{#1}^{\canonical}}

\newcommand\rel[1]{\mathbb{#1}}
\newcommand\alg[1]{\mathbf{#1}}
\newcommand\minabs{\mathrel{\lhd\kern-2pt\lhd}}

\newcommand\Projs{\mathscr P}
\newcommand\gG{\mathscr G}
\newcommand\cC{\mathscr C}
\newcommand\cD{\mathscr D}
\newcommand\Q{\mathbb Q}
\DeclareMathOperator{\Proj}{\Projs}

\DeclareMathOperator{\fpp}{FPP}
\DeclareMathOperator{\lex}{lex}

\newcommand\proj[2]{\pro_{#1}{(#2)}}

\DeclareMathOperator{\pro}{proj}

\DeclareMathOperator\Csp{CSP}
\DeclareMathOperator{\Pol}{Pol}

\DeclareMathOperator{\Aut}{Aut}

\newcommand\actson{\curvearrowright}

\DeclareMathOperator{\V}{\mathcal{V}}

\DeclareMathOperator{\constraints}{\mathcal{C}}
\DeclareMathOperator{\instance}{\mathcal{I}}
\DeclareMathOperator{\cJ}{\mathcal{J}}
\DeclareMathOperator{\clone}{\mathscr{C}}
\DeclareMathOperator{\group}{\mathscr{G}}

\newcounter{bwcounter}

%% file: main.bbl
\begin{thebibliography}{10}

\bibitem{Dawar:2017}
Samson Abramsky, Anuj Dawar, and Pengming Wang.
\newblock The pebbling comonad in finite model theory.
\newblock In {\em 32nd Annual {ACM/IEEE} Symposium on Logic in Computer Science, {LICS} 2017, Reykjavik, Iceland, June 20-23, 2017}, pages 1--12. {IEEE} Computer Society, 2017.
\newblock \href {https://doi.org/10.1109/LICS.2017.8005129} {\path{doi:10.1109/LICS.2017.8005129}}.

\bibitem{Atserias:2007}
Albert Atserias, Andrei~A. Bulatov, and V{\'{\i}}ctor Dalmau.
\newblock On the power of \emph{k} -consistency.
\newblock In Lars Arge, Christian Cachin, Tomasz Jurdzinski, and Andrzej Tarlecki, editors, {\em Automata, Languages and Programming, 34th International Colloquium, {ICALP} 2007, Wroclaw, Poland, July 9-13, 2007, Proceedings}, volume 4596 of {\em Lecture Notes in Computer Science}, pages 279--290. Springer, 2007.
\newblock \href {https://doi.org/10.1007/978-3-540-73420-8\_26} {\path{doi:10.1007/978-3-540-73420-8\_26}}.

\bibitem{Babai:2016}
L{\'{a}}szl{\'{o}} Babai.
\newblock Graph isomorphism in quasipolynomial time [extended abstract].
\newblock In Daniel Wichs and Yishay Mansour, editors, {\em Proceedings of the 48th Annual {ACM} {SIGACT} Symposium on Theory of Computing, {STOC} 2016, Cambridge, MA, USA, June 18-21, 2016}, pages 684--697. {ACM}, 2016.
\newblock \href {https://doi.org/10.1145/2897518.2897542} {\path{doi:10.1145/2897518.2897542}}.

\bibitem{BartoCollapse}
Libor Barto.
\newblock The collapse of the bounded width hierarchy.
\newblock {\em J. Log. Comput.}, 26(3):923--943, 2016.
\newblock \href {https://doi.org/10.1093/logcom/exu070} {\path{doi:10.1093/logcom/exu070}}.

\bibitem{BKOPP}
Libor Barto, Michael Kompatscher, Miroslav Ol\v{s}\'{a}k, Trung~Van Pham, and Michael Pinsker.
\newblock The equivalence of two dichotomy conjectures for infinite domain constraint satisfaction problems.
\newblock In {\em Proceedings of the 32nd Annual {ACM/IEEE} Symposium on Logic in Computer Science -- LICS'17}, 2017.
\newblock Preprint arXiv:1612.07551.

\bibitem{BKOPP-equations}
Libor Barto, Michael Kompatscher, Miroslav Ol\v{s}\'{a}k, Trung~Van Pham, and Michael Pinsker.
\newblock Equations in oligomorphic clones and the constraint satisfaction problem for $\omega$-categorical structures.
\newblock {\em Journal of Mathematical Logic}, 19(2):\#1950010, 2019.

\bibitem{BartoKozikCyclic}
Libor Barto and Marcin Kozik.
\newblock Absorbing subalgebras, cyclic terms, and the constraint satisfaction problem.
\newblock {\em Log. Methods Comput. Sci.}, 8(1), 2012.
\newblock \href {https://doi.org/10.2168/LMCS-8(1:7)2012} {\path{doi:10.2168/LMCS-8(1:7)2012}}.

\bibitem{BartoKozikBoundedWidth}
Libor Barto and Marcin Kozik.
\newblock Constraint satisfaction problems solvable by local consistency methods.
\newblock {\em J. {ACM}}, 61(1):3:1--3:19, 2014.
\newblock \href {https://doi.org/10.1145/2556646} {\path{doi:10.1145/2556646}}.

\bibitem{wonderland}
Libor Barto, Jakub Opr{\v s}al, and Michael Pinsker.
\newblock The wonderland of reflections.
\newblock {\em Israel Journal of Mathematics}, 223(1):363--398, 2018.

\bibitem{BartoPinskerDichotomy}
Libor Barto and Michael Pinsker.
\newblock The algebraic dichotomy conjecture for infinite domain constraint satisfaction problems.
\newblock In {\em Proceedings of the 31th {A}nnual {IEEE} {S}ymposium on {L}ogic in {C}omputer {S}cience -- {LICS}'16}, pages 615--622, 2016.
\newblock Preprint arXiv:1602.04353.

\bibitem{Topo}
Libor Barto and Michael Pinsker.
\newblock Topology is irrelevant.
\newblock {\em SIAM Journal on Computing}, 49(2):365--393, 2020.

\bibitem{Bienvenu:2014}
Meghyn Bienvenu, Balder ten Cate, Carsten Lutz, and Frank Wolter.
\newblock Ontology-based data access: {A} study through disjunctive datalog, {CSP}, and {MMSNP}.
\newblock {\em {ACM} Trans. Database Syst.}, 39(4):33:1--33:44, 2014.
\newblock \href {https://doi.org/10.1145/2661643} {\path{doi:10.1145/2661643}}.

\bibitem{cores}
Manuel Bodirsky.
\newblock {Cores of Countably Categorical Structures}.
\newblock {\em {Logical Methods in Computer Science}}, {Volume 3, Issue 1}, January 2007.
\newblock \href {https://doi.org/10.2168/LMCS-3(1:2)2007} {\path{doi:10.2168/LMCS-3(1:2)2007}}.

\bibitem{BodirskyRamsey}
Manuel Bodirsky.
\newblock Ramsey classes: Examples and constructions.
\newblock In {\em Surveys in Combinatorics. London Mathematical Society Lecture Note Series 424}. Cambridge University Press, 2015.
\newblock Invited survey article for the British Combinatorial Conference; ArXiv:1502.05146.

\bibitem{BodirskyBook}
Manuel Bodirsky.
\newblock {\em Complexity of Infinite-Domain Constraint Satisfaction}.
\newblock Lecture Notes in Logic. Cambridge University Press, 2021.
\newblock \href {https://doi.org/10.1017/9781107337534} {\path{doi:10.1017/9781107337534}}.

\bibitem{BodChenPinsker}
Manuel Bodirsky, Hubie Chen, and Michael Pinsker.
\newblock The reducts of equality up to primitive positive interdefinability.
\newblock {\em Journal of Symbolic Logic}, 75(4):1249--1292, 2010.

\bibitem{BodirskyDalmau}
Manuel Bodirsky and V{\'{\i}}ctor Dalmau.
\newblock Datalog and constraint satisfaction with infinite templates.
\newblock {\em J. Comput. Syst. Sci.}, 79(1):79--100, 2013.
\newblock A conference version appeared in the Proceedings of the 34th Symposium on Theoretical Aspects of Computer Science (STACS 2006), pages 646--659.
\newblock \href {https://doi.org/10.1016/j.jcss.2012.05.012} {\path{doi:10.1016/j.jcss.2012.05.012}}.

\bibitem{BodirskyGrohe}
Manuel Bodirsky and Martin Grohe.
\newblock Non-dichotomies in constraint satisfaction complexity.
\newblock In Luca Aceto, Ivan Damg{\aa}rd, Leslie~Ann Goldberg, Magn{\'{u}}s~M. Halld{\'{o}}rsson, Anna Ing{\'{o}}lfsd{\'{o}}ttir, and Igor Walukiewicz, editors, {\em Automata, Languages and Programming, 35th International Colloquium, {ICALP} 2008, Reykjavik, Iceland, July 7-11, 2008, Proceedings, Part {II} - Track {B:} Logic, Semantics, and Theory of Programming {\&} Track {C:} Security and Cryptography Foundations}, volume 5126 of {\em Lecture Notes in Computer Science}, pages 184--196. Springer, 2008.
\newblock \href {https://doi.org/10.1007/978-3-540-70583-3\_16} {\path{doi:10.1007/978-3-540-70583-3\_16}}.

\bibitem{ecsps}
Manuel Bodirsky and Jan K\'ara.
\newblock The complexity of equality constraint languages.
\newblock {\em Theory of Computing Systems}, 3(2):136--158, 2008.
\newblock A conference version appeared in the proceedings of Computer Science Russia {(CSR'06)}.

\bibitem{MMSNP}
Manuel Bodirsky, Florent~R. Madelaine, and Antoine Mottet.
\newblock A universal-algebraic proof of the complexity dichotomy for {M}onotone {M}onadic {SNP}.
\newblock In Anuj Dawar and Erich Gr{\"{a}}del, editors, {\em Proceedings of the 33rd Annual {ACM/IEEE} Symposium on Logic in Computer Science, {LICS} 2018, Oxford, UK, July 09-12, 2018}, pages 105--114. {ACM}, 2018.
\newblock \href {https://doi.org/10.1145/3209108.3209156} {\path{doi:10.1145/3209108.3209156}}.

\bibitem{MMSNP-journal}
Manuel Bodirsky, Florent~R. Madelaine, and Antoine Mottet.
\newblock {A Proof of the Algebraic Tractability Conjecture for {M}onotone {M}onadic {SNP}}.
\newblock {\em SIAM Journal on Computing}, 50(4):1359--1409, 2021.
\newblock \href {https://doi.org/10.1137/19M128466X} {\path{doi:10.1137/19M128466X}}.

\bibitem{Mottet:2018}
Manuel Bodirsky, Barnaby Martin, and Antoine Mottet.
\newblock Discrete temporal constraint satisfaction problems.
\newblock {\em J. {ACM}}, 65(2):9:1--9:41, 2018.
\newblock \href {https://doi.org/10.1145/3154832} {\path{doi:10.1145/3154832}}.

\bibitem{HomogeneousGraphs}
Manuel Bodirsky, Barnaby Martin, Michael Pinsker, and Andr{\'{a}}s Pongr{\'{a}}cz.
\newblock Constraint satisfaction problems for reducts of homogeneous graphs.
\newblock {\em {SIAM} J. Comput.}, 48(4):1224--1264, 2019.
\newblock A conference version appeared in the Proceedings of the 43rd International Colloquium on Automata, Languages, and Programming, {ICALP} 2016, pages 119:1-119:14.
\newblock \href {https://doi.org/10.1137/16M1082974} {\path{doi:10.1137/16M1082974}}.

\bibitem{ReductionFinite}
Manuel Bodirsky and Antoine Mottet.
\newblock Reducts of finitely bounded homogeneous structures, and lifting tractability from finite-domain constraint satisfaction.
\newblock In Martin Grohe, Eric Koskinen, and Natarajan Shankar, editors, {\em Proceedings of the 31st Annual {ACM/IEEE} Symposium on Logic in Computer Science, {LICS} '16, New York, NY, USA, July 5-8, 2016}, pages 623--632. {ACM}, 2016.
\newblock \href {https://doi.org/10.1145/2933575.2934515} {\path{doi:10.1145/2933575.2934515}}.

\bibitem{ReductsUnary}
Manuel Bodirsky and Antoine Mottet.
\newblock A dichotomy for first-order reducts of unary structures.
\newblock {\em Log. Methods Comput. Sci.}, 14(2), 2018.
\newblock \href {https://doi.org/10.23638/LMCS-14(2:13)2018} {\path{doi:10.23638/LMCS-14(2:13)2018}}.

\bibitem{TopoRelevant}
Manuel Bodirsky, Antoine Mottet, Miroslav Ol{\v {s}}{\'{a}}k, Jakub Opr{\v {s}}al, Michael Pinsker, and Ross Willard.
\newblock Topology is relevant (in a dichotomy conjecture for infinite-domain constraint satisfaction problems).
\newblock In {\em 34th Annual {ACM/IEEE} Symposium on Logic in Computer Science, {LICS} 2019, Vancouver, BC, Canada, June 24-27, 2019}, pages 1--12. {IEEE}, 2019.
\newblock \href {https://doi.org/10.1109/LICS.2019.8785883} {\path{doi:10.1109/LICS.2019.8785883}}.

\bibitem{TopoRelevant-TAMS}
Manuel Bodirsky, Antoine Mottet, Miroslav Olšák, Jakub Opršal, Michael Pinsker, and Ross Willard.
\newblock $\omega$-categorical structures avoiding height 1 identities.
\newblock {\em {Transactions of the AMS}}, 374:327--350, 2021.
\newblock \href {https://doi.org/10.1090/tran/8179} {\path{doi:10.1090/tran/8179}}.

\bibitem{BodNesetril}
Manuel Bodirsky and Jaroslav Ne{\v{s}}et{\v{r}}il.
\newblock Constraint satisfaction with countable homogeneous templates.
\newblock In Matthias Baaz and Johann~A. Makowsky, editors, {\em Computer Science Logic}, pages 44--57, Berlin, Heidelberg, 2003. Springer Berlin Heidelberg.

\bibitem{Rydval:2020}
Manuel Bodirsky, Wied Pakusa, and Jakub Rydval.
\newblock Temporal constraint satisfaction problems in fixed-point logic.
\newblock In Holger Hermanns, Lijun Zhang, Naoki Kobayashi, and Dale Miller, editors, {\em {LICS} '20: 35th Annual {ACM/IEEE} Symposium on Logic in Computer Science, Saarbr{\"{u}}cken, Germany, July 8-11, 2020}, pages 237--251. {ACM}, 2020.
\newblock \href {https://doi.org/10.1145/3373718.3394750} {\path{doi:10.1145/3373718.3394750}}.

\bibitem{Schaefer-Graphs}
Manuel Bodirsky and Michael Pinsker.
\newblock Schaefer's theorem for graphs.
\newblock {\em J. {ACM}}, 62(3):19:1--19:52, 2015.
\newblock A conference version appeared in the Proceedings of STOC 2011, pages 655-664.
\newblock \href {https://doi.org/10.1145/2764899} {\path{doi:10.1145/2764899}}.

\bibitem{BodPin-CanonicalFunctions}
Manuel Bodirsky and Michael Pinsker.
\newblock {C}anonical {F}unctions: a {P}roof via {T}opological {D}ynamics.
\newblock {\em Homogeneous Structures, A Workshop in Honour of Norbert Sauer's 70th Birthday, Contributions to Discrete Mathematics}, 16(2):36--45, 2021.

\bibitem{BPP-projective-homomorphisms}
Manuel Bodirsky, Michael Pinsker, and Andr\'{a}s Pongr\'acz.
\newblock Projective clone homomorphisms.
\newblock {\em Journal of Symbolic Logic}, 86(1):148–161, 2021.
\newblock \href {https://doi.org/10.1017/jsl.2019.23} {\path{doi:10.1017/jsl.2019.23}}.

\bibitem{BPT-decidability-of-definability}
Manuel Bodirsky, Michael Pinsker, and Todor Tsankov.
\newblock Decidability of definability.
\newblock {\em Journal of Symbolic Logic}, 78(4):1036--1054, 2013.
\newblock A conference version appeared in the Proceedings of the Twenty-Sixth Annual IEEE Symposium on Logic in Computer Science ({LICS}) 2011, pages {321--328}.

\bibitem{BodirskyW12}
Manuel Bodirsky and Micha\l\ Wrona.
\newblock Equivalence constraint satisfaction problems.
\newblock In {\em Computer Science Logic (CSL'12) - 26th International Workshop/21st Annual Conference of the EACSL, {CSL} 2012, September 3-6, 2012, Fontainebleau, France}, pages 122--136, 2012.
\newblock \href {https://doi.org/10.4230/LIPIcs.CSL.2012.122} {\path{doi:10.4230/LIPIcs.CSL.2012.122}}.

\bibitem{BourhisLutz}
Pierre Bourhis and Carsten Lutz.
\newblock Containment in monadic disjunctive datalog, {MMSNP}, and expressive description logics.
\newblock In Chitta Baral, James~P. Delgrande, and Frank Wolter, editors, {\em Principles of Knowledge Representation and Reasoning: Proceedings of the Fifteenth International Conference, {KR} 2016, Cape Town, South Africa, April 25-29, 2016}, pages 207--216. {AAAI} Press, 2016.
\newblock URL: \url{http://www.aaai.org/ocs/index.php/KR/KR16/paper/view/12847}.

\bibitem{NebelBueckert}
Hans-J{\"u}rgen B{\"u}ckert and Bernhard Nebel.
\newblock Reasoning about temporal relations: a maximal tractable subclass of {A}llen’s interval algebra.
\newblock {\em J. {ACM}}, 42(1):43--66, 1995.

\bibitem{Bulatov:2017}
Andrei~A. Bulatov.
\newblock A dichotomy theorem for nonuniform {CSP}s.
\newblock In Chris Umans, editor, {\em 58th {IEEE} Annual Symposium on Foundations of Computer Science, {FOCS} 2017, Berkeley, CA, USA, October 15-17, 2017}, pages 319--330. {IEEE} Computer Society, 2017.
\newblock \href {https://doi.org/10.1109/FOCS.2017.37} {\path{doi:10.1109/FOCS.2017.37}}.

\bibitem{DualitiesCSP}
Andrei~A. Bulatov, Andrei~A. Krokhin, and Beno{\^{\i}}t Larose.
\newblock Dualities for constraint satisfaction problems.
\newblock In Nadia Creignou, Phokion~G. Kolaitis, and Heribert Vollmer, editors, {\em Complexity of Constraints - An Overview of Current Research Themes [Result of a Dagstuhl Seminar]}, volume 5250 of {\em Lecture Notes in Computer Science}, pages 93--124. Springer, 2008.
\newblock \href {https://doi.org/10.1007/978-3-540-92800-3\_5} {\path{doi:10.1007/978-3-540-92800-3\_5}}.

\bibitem{MetaChenLarose}
Hubie Chen and Beno{\^{\i}}t Larose.
\newblock Asking the metaquestions in constraint tractability.
\newblock {\em {ACM} Trans. Comput. Theory}, 9(3):11:1--11:27, 2017.
\newblock \href {https://doi.org/10.1145/3134757} {\path{doi:10.1145/3134757}}.

\bibitem{HubickaNesetril}
Jan~Hubi\v cka and Jaroslav~Ne\v set\v ril.
\newblock All those ramsey classes (ramsey classes with closures and forbidden homomorphisms).
\newblock {\em Advances in Mathematics}, 356:106791, 2019.
\newblock \href {https://doi.org/10.1016/j.aim.2019.106791} {\path{doi:10.1016/j.aim.2019.106791}}.

\bibitem{DalmauPearson}
V\'ictor Dalmau and Justin Pearson.
\newblock Closure functions and width 1 problems.
\newblock In {\em Proceedings of the International Conference on Principles and Practice of Constraint Programming (CP)}, pages 159--173, 1999.

\bibitem{FederVardi}
Tom{\'{a}}s Feder and Moshe~Y. Vardi.
\newblock The computational structure of monotone monadic {SNP} and constraint satisfaction: {A} study through datalog and group theory.
\newblock {\em {SIAM} J. Comput.}, 28(1):57--104, 1998.
\newblock \href {https://doi.org/10.1137/S0097539794266766} {\path{doi:10.1137/S0097539794266766}}.

\bibitem{FeierKuusistoLutz}
Cristina Feier, Antti Kuusisto, and Carsten Lutz.
\newblock Rewritability in monadic disjunctive datalog, {MMSNP}, and expressive description logics.
\newblock {\em Log. Methods Comput. Sci.}, 15(2), 2019.
\newblock A conference version appeared in the Proceedings of the 20th International Conference on Database Theory (ICDT17) 2017, pages 1:1--1:17.
\newblock \href {https://doi.org/10.23638/LMCS-15(2:15)2019} {\path{doi:10.23638/LMCS-15(2:15)2019}}.

\bibitem{Hrushovski-monsters}
Pierre Gillibert, Julius Jonu{\v {s}}as, Michael Kompatscher, Antoine Mottet, and Michael Pinsker.
\newblock Hrushovski's encoding and {\(\omega\)}-categorical {CSP} monsters.
\newblock In Artur Czumaj, Anuj Dawar, and Emanuela Merelli, editors, {\em 47th International Colloquium on Automata, Languages, and Programming, {ICALP} 2020, July 8-11, 2020, Saarbr{\"{u}}cken, Germany (Virtual Conference)}, volume 168 of {\em LIPIcs}, pages 131:1--131:17. Schloss Dagstuhl - Leibniz-Zentrum f{\"{u}}r Informatik, 2020.
\newblock \href {https://doi.org/10.4230/LIPIcs.ICALP.2020.131} {\path{doi:10.4230/LIPIcs.ICALP.2020.131}}.

\bibitem{Grohe:2012}
Martin Grohe.
\newblock Fixed-point definability and polynomial time on graphs with excluded minors.
\newblock {\em J. {ACM}}, 59(5):27:1--27:64, 2012.
\newblock \href {https://doi.org/10.1145/2371656.2371662} {\path{doi:10.1145/2371656.2371662}}.

\bibitem{CSPs_RQ}
Peter Jonsson and Johan Thapper.
\newblock Constraint satisfaction and semilinear expansions of addition over the rationals and the reals.
\newblock {\em Journal of Computer and System Sciences}, 82(5):912--928, 2016.
\newblock URL: \url{https://www.sciencedirect.com/science/article/pii/S0022000016000271}, \href {https://doi.org/10.1016/j.jcss.2016.03.002} {\path{doi:10.1016/j.jcss.2016.03.002}}.

\bibitem{Topo-Dynamics}
Alexander Kechris, Vladimir Pestov, and Stevo Todor\v{c}evi\'c.
\newblock Fra\"{i}ss\'e limits, {R}amsey theory, and topological dynamics of automorphism groups.
\newblock {\em Geometric and Functional Analysis}, 15(1):106--189, 2005.

\bibitem{Kolaitis:2000}
Phokion~G. Kolaitis and Moshe~Y. Vardi.
\newblock Conjunctive-query containment and constraint satisfaction.
\newblock {\em J. Comput. Syst. Sci.}, 61(2):302--332, 2000.
\newblock A conference version appeared in the Proceedings of Symposium on Principles of Database Systems (PODS) 1998, pages 205-213.
\newblock \href {https://doi.org/10.1006/jcss.2000.1713} {\path{doi:10.1006/jcss.2000.1713}}.

\bibitem{PosetCSP}
Michael Kompatscher and Trung~Van Pham.
\newblock A complexity dichotomy for poset constraint satisfaction.
\newblock {\em {FLAP}}, 5(8):1663--1696, 2018.
\newblock A conference version appeared in the Proceedings of the 34th Symposium on Theoretical Aspects of Computer Science (STACS 2017), pages 47:1–47:12.
\newblock URL: \url{https://www.collegepublications.co.uk/downloads/ifcolog00028.pdf}.

\bibitem{Kozik:2016}
Marcin Kozik.
\newblock Weak consistency notions for all the {CSP}s of bounded width.
\newblock In Martin Grohe, Eric Koskinen, and Natarajan Shankar, editors, {\em Proceedings of the 31st Annual {ACM/IEEE} Symposium on Logic in Computer Science, {LICS} '16, New York, NY, USA, July 5-8, 2016}, pages 633--641. {ACM}, 2016.
\newblock \href {https://doi.org/10.1145/2933575.2934510} {\path{doi:10.1145/2933575.2934510}}.

\bibitem{Maltsev-Cond}
Marcin Kozik, Andrei Krokhin, Matt Valeriote, and Ross Willard.
\newblock Characterizations of several {Maltsev} conditions.
\newblock {\em Algebra universalis}, 73(3-4):205--224, 2015.
\newblock \href {https://doi.org/10.1007/s00012-015-0327-2} {\path{doi:10.1007/s00012-015-0327-2}}.

\bibitem{LaroseZadori}
Benoit Larose and L{\'a}szl{\'o} Z\'adori.
\newblock Bounded width problems and algebras.
\newblock {\em Algebra Universalis}, 56(3-4):439--466, 2007.

\bibitem{MarotiMcKenzie}
Mikl\'os Mar\'oti and Ralph McKenzie.
\newblock Existence theorems for weakly symmetric operations.
\newblock {\em Algebra Universalis}, 59(3), 2008.

\bibitem{hypergraphs}
Antoine Mottet, Tom\'a\v{s} Nagy, and Michael Pinsker.
\newblock An order out of nowhere: a new algorithm for infinite-domain csps, 2023.
\newblock \href {https://arxiv.org/abs/2301.12977} {\path{arXiv:2301.12977}}.

\bibitem{SmoothApproximations}
Antoine Mottet and Michael Pinsker.
\newblock Smooth approximations and {CSP}s over finitely bounded homogeneous structures.
\newblock In {\em Proceedings of the 37th Annual ACM/IEEE Symposium on Logic in Computer Science}, LICS '22, New York, NY, USA, 2022. Association for Computing Machinery.
\newblock \href {https://doi.org/10.1145/3531130.3533353} {\path{doi:10.1145/3531130.3533353}}.

\bibitem{strictwidth}
Tom\'{a}\v{s} Nagy and Michael Pinsker.
\newblock Strict width for constraint satisfaction problems over homogeneous strucures of finite duality, 2024.
\newblock \href {https://arxiv.org/abs/2402.09951} {\path{arXiv:2402.09951}}.

\bibitem{infinitesheep}
Michael Pinsker.
\newblock Current challenges in infinite-domain constraint satisfaction: Dilemmas of the infinite sheep.
\newblock In {\em 2022 IEEE 52nd International Symposium on Multiple-Valued Logic (ISMVL)}, pages 80--87, Los Alamitos, CA, USA, may 2022. IEEE Computer Society.
\newblock URL: \url{https://doi.ieeecomputersociety.org/10.1109/ISMVL52857.2022.00019}, \href {https://doi.org/10.1109/ISMVL52857.2022.00019} {\path{doi:10.1109/ISMVL52857.2022.00019}}.

\bibitem{Pseudo-loop}
Michael Pinsker, Pierre Gillibert, and Julius Jonu\v{s}as.
\newblock Pseudo-loop conditions.
\newblock {\em Bulletin of the London Mathematical Society}, 51(5):917--936, 2019.

\bibitem{Valeriote}
Matthew~A. Valeriote.
\newblock A subalgebra intersection property for congruence distributive varieties.
\newblock {\em Canadian Journal of Mathematics}, 61(2):451–464, 2009.
\newblock \href {https://doi.org/10.4153/CJM-2009-023-2} {\path{doi:10.4153/CJM-2009-023-2}}.

\bibitem{Wrona:2020b}
Micha\l\ Wrona.
\newblock On the relational width of first-order expansions of finitely bounded homogeneous binary cores with bounded strict width.
\newblock In Holger Hermanns, Lijun Zhang, Naoki Kobayashi, and Dale Miller, editors, {\em {LICS} '20: 35th Annual {ACM/IEEE} Symposium on Logic in Computer Science, Saarbr{\"{u}}cken, Germany, July 8-11, 2020}, pages 958--971. {ACM}, 2020.
\newblock \href {https://doi.org/10.1145/3373718.3394781} {\path{doi:10.1145/3373718.3394781}}.

\bibitem{Wrona:2020a}
Micha\l\ Wrona.
\newblock Relational width of first-order expansions of homogeneous graphs with bounded strict width.
\newblock In Christophe Paul and Markus Bl{\"{a}}ser, editors, {\em 37th International Symposium on Theoretical Aspects of Computer Science, {STACS} 2020, March 10-13, 2020, Montpellier, France}, volume 154 of {\em LIPIcs}, pages 39:1--39:16. Schloss Dagstuhl - Leibniz-Zentrum f{\"{u}}r Informatik, 2020.
\newblock \href {https://doi.org/10.4230/LIPIcs.STACS.2020.39} {\path{doi:10.4230/LIPIcs.STACS.2020.39}}.

\bibitem{Wrona:2024}
Micha{\l} Wrona.
\newblock Quasi directed jonsson operations imply bounded width (for fo-expansions of symmetric binary cores with free amalgamation), 2024.
\newblock \href {https://arxiv.org/abs/2402.16540} {\path{arXiv:2402.16540}}.

\bibitem{Zhuk:2017}
Dmitriy Zhuk.
\newblock A proof of {CSP} dichotomy conjecture.
\newblock In Chris Umans, editor, {\em 58th {IEEE} Annual Symposium on Foundations of Computer Science, {FOCS} 2017, Berkeley, CA, USA, October 15-17, 2017}, pages 331--342. {IEEE} Computer Society, 2017.
\newblock \href {https://doi.org/10.1109/FOCS.2017.38} {\path{doi:10.1109/FOCS.2017.38}}.

\bibitem{Zhuk:2020}
Dmitriy Zhuk.
\newblock A proof of the {CSP} dichotomy conjecture.
\newblock {\em J. {ACM}}, 67(5):30:1--30:78, 2020.
\newblock \href {https://doi.org/10.1145/3402029} {\path{doi:10.1145/3402029}}.

\bibitem{StrongSubalgebras}
Dmitriy Zhuk.
\newblock Strong subalgebras and the constraint satisfaction problem.
\newblock {\em J. Multiple Valued Log. Soft Comput.}, 36(4-5):455--504, 2021.

\end{thebibliography}
